\documentclass[aps,pra,twocolumn,superscriptaddress,nofootinbib,notitlepage,10pt]{revtex4-2}
\usepackage{epsfig,amssymb,amsmath,mathrsfs,amsfonts,color,amsthm,graphicx,float,color,bm,url}
\usepackage[colorlinks=true]{hyperref}

\usepackage[utf8,latin1]{inputenc}               %Input what you want e.g.
\usepackage[T1]{fontenc}                           %Output what you want e.g.
\usepackage[british]{babel}                      %Do hyphenation according to british english
\usepackage{mathpazo}      

\usepackage{graphicx,color}
\usepackage[table]{xcolor}
\usepackage{comment}

\usepackage{array}
\newcolumntype{P}[1]{>{\centering\arraybackslash}p{#1}}
\newcolumntype{M}[1]{>{\centering\arraybackslash}m{#1}}

\usepackage{algpseudocode}
\usepackage[active]{srcltx}
\usepackage{subfigure}
\def\<{\langle}\def\>{\rangle}

\def\:{\hbox{\bf
    :}}

\def\spc#1{\mathcal{#1}}
\def\citationorder{0}

\newtheorem{theo}{{Theorem}}
\newtheorem{defi}{{Definition}}
\newtheorem{lemma}{{Lemma}}

\newcommand{\h}[1]{\mathcal{H}_{#1}}
\newcommand{\im}{\mathring \imath} % custom imaginary unit
\usepackage{calc} % for SVG images created with Inkscape
\graphicspath{{Images/}} % for SVG images created with Inkscape
\usepackage{dsfont} % for identity \mathds{1}
\usepackage{physics} % for ket, bra, ...
\usepackage{qcircuit} % for quantum circuit
\usepackage{blindtext} % for lorem ipsum in abstract

\usepackage[linesnumbered,ruled,vlined]{algorithm2e} % for algorithm

\newcommand{\lket}[1]{\mathinner{|{#1}\rangle}}
\newcommand{\lbra}[1]{\mathinner{ \langle{#1}|}}

\begin{document}
\title{Quantum Metrology for Non-Markovian Processes}

\author{Anian Altherr}
\email{aaltherr@ethz.ch}
\affiliation{Institute for Theoretical Physics, ETH Z\"urich, 8093 Z\"urich, Switzerland}
\author{Yuxiang Yang}
\email{yangyu@ethz.ch}
\affiliation{Institute for Theoretical Physics, ETH Z\"urich, 8093 Z\"urich, Switzerland} 
\affiliation{QICI Quantum Information and Computation Initiative, Department of Computer Science, The University of Hong Kong, Pokfulam Road, Hong Kong} 

\begin{abstract}
Quantum metrology is a rapidly developing branch of quantum technologies. While various theories have been established on quantum metrology for Markovian processes, i.e., quantum channel estimation, quantum metrology for non-Markovian processes is much less explored.
In this Letter, we establish a general framework of non-Markovian quantum metrology. For any parametrized non-Markovian process on a finite-dimensional system, we derive a formula for the maximal amount of quantum Fisher information that can be extracted from it by  an optimally controlled probe state. In addition, we design an algorithm that evaluates this quantum Fisher information via semidefinite programming.
We apply our framework to noisy frequency estimation, where we find that the optimal performance of quantum metrology is better in the non-Markovian scenario than in the Markovian scenario and  explore the possibility of efficient sensing via simple variational circuits.
\end{abstract}

\maketitle
\medskip

\noindent{\em Introduction.}   
	Quantum metrology holds the promise of an early application of quantum technologies that offer an advantage over classical ones. 
	This advantage, nevertheless, is often sensitive to noise \cite{Huelga1997,Escher2011,Demkowicz2012,Kolodynski2013}. Advances in quantum metrology have been focusing on identifying the ultimate limit of parameter estimation in the presence of noise \cite{Escher2011,Demkowicz2012,Smirne2016,yuan2017quantum,zhou2020asymptotic}. 
	So far, the analysis has mostly been constrained to the Markovian setting and carried out within the well-established model of quantum channel estimation [Fig. \ref{fig:introduction} (a)], where various tools have been developed. In particular, it has recently been shown that the quantum advantage can be retrieved by using appropriate quantum control \cite{Kessler2014, Dur2014,Demkowicz2017,Zhou2018}.

	Non-Markovian quantum metrology [Fig. \ref{fig:introduction} (b)], on the other hand, is much less explored. 
	Despite interesting findings in  specific cases (see, e.g., Ref.~\cite{Matsuzaki2011,Chin2012,Berrada2013,macieszczak2015zeno,Wang2017,Yang2019,mirkin2020quantum}), there has  not been a systematic way to treat non-Markovian quantum metrology. As a first step, in Ref.~\cite{Yang2019}, one of us proposed a general framework for non-Markovian quantum metrology based on an information-theoretic structure, named the quantum comb \cite{Chiribella2007,Chiribella2008,Chiribella2009}, which has also been applied to a broad range of tasks including channel discrimination \cite{Chiribella2008Metrology}, quantum network optimization \cite{Chiribella2016}, Markovianity \cite{taranto2019quantum}, and probabilistic quantum computation \cite{dong2021success}.
    However, given an arbitrary non-Markovian metrology task, the core question of determining \emph{the ultimate precision limit}, optimized over all possible means of quantum probe preparation and control, still remains open. This issue is pressing not only because non-Markovian processes are prevalent in physics, but also for the rapid development of quantum computing:
	With the advance of NISQ (near-term intermediate-scale quantum) devices \cite{preskill2018quantum}, deeper quantum circuits will appear where memory effects of the environment are more significant. New techniques will soon be needed to test and benchmark complex noises generated by such effects.
	
	In this Letter, we establish a general framework of assessing the precision of non-Markovian quantum metrology, quantified by the quantum Fisher information (QFI). The framework is built upon our preceding work~\cite{Yang2019}, where non-Markovian metrology is modelled using quantum combs. Here we address the pivotal problem of evaluating the QFI of arbitrary quantum combs, deriving both a general formula and an algorithm that efficiently computes the QFI of arbitrary combs via semidefinite programming.  As a working example, we apply the algorithm to the task of frequency estimation under non-Markovian noise, where we observe that the QFI can be increased by applying suitable quantum control to the system. We also show that, when the memory of the non-Markovian noise is not too deep, this control can be well approximated by a simple variational circuit, whose complexity grows only linearly in the memory depth. Our results extend  quantum metrology to generic physical processes with memory, paving the way for various future research and applications.

\medskip

\noindent{\it Quantum Fisher information of quantum combs.}
Our goal is to estimate a single parameter $\theta$ from a non-Markovian process that carries this parameter. For instance, the task could be to measure the frequency discrepancy in a noisy atomic clock where the atoms are coupled to a persistent environment, or to estimate the strength of time-correlated noise in a deep quantum circuit.

	\begin{figure}[t]	
		\centering
		\includegraphics[width=\columnwidth]{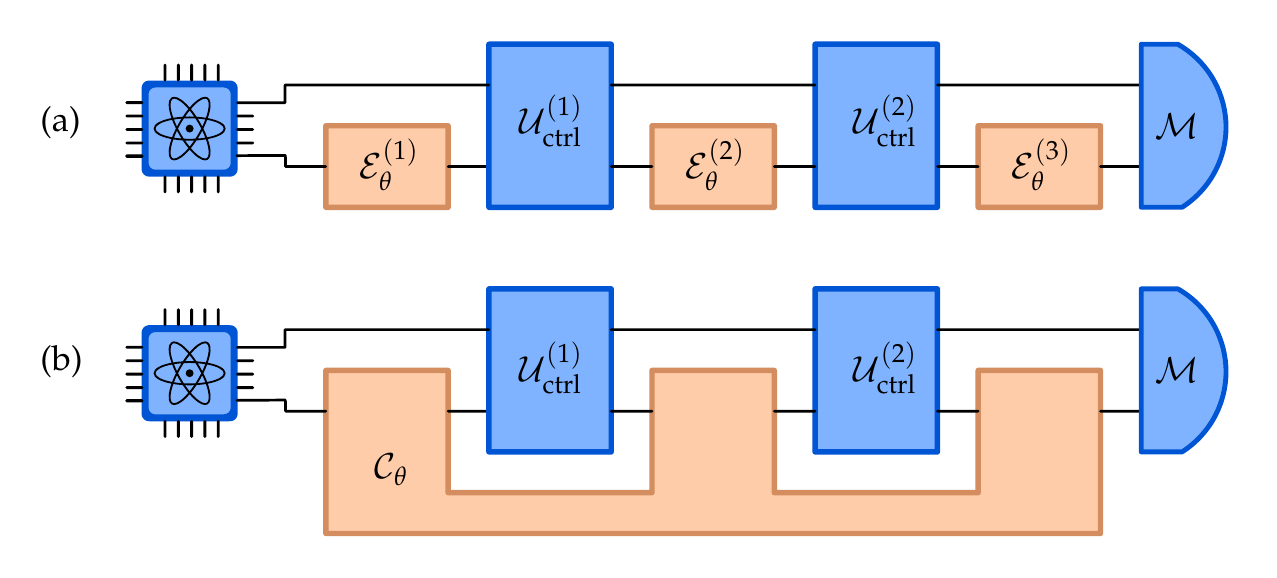}
		% The computer graphics is from https://freesvg.org/computer-system-1573730293.
		% The buttons are from the same page.
		\caption{{\bf Markovian versus non-Markovian quantum metrology.} (a): The standard, Markovian setting of quantum metrology where the goal is to estimate $\theta$ from a sequence of quantum channels (orange), and the general strategy is to prepare a probe state and apply adaptive quantum control (blue). (b): Non-Markovian quantum metrology where $\theta$ is encoded in a non-Markovian process with inaccessible memory (orange).}
		\label{fig:introduction}
		\end{figure}
		
			\begin{figure}[b]
		\centering
		\includegraphics[width=\columnwidth]{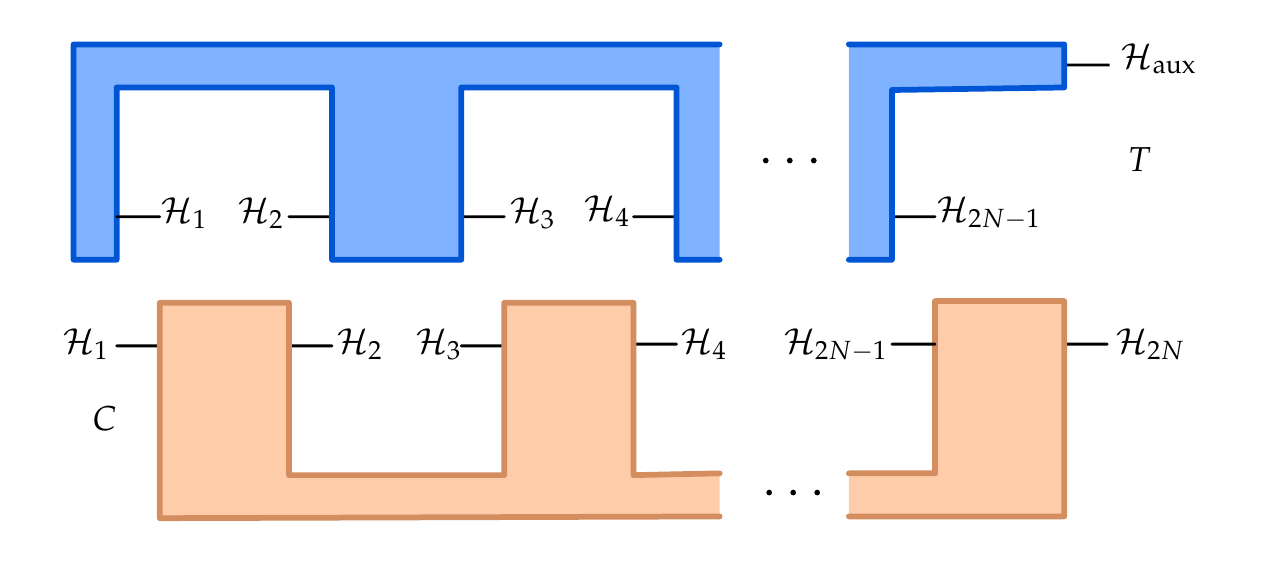}
		\caption{{\bf Quantum combs and probes.} A generic quantum comb $C \in \text{Comb}[ (\h{1},\h{2}), \ldots, (\h{2N-1},\h{2N})]$ (orange) consists of $N$ teeth, each representing one time step. 
		%At each time step a quantum channel is performed, whose output state is partially sent out of the comb and partially stored in the underlying memory.
		A probe $T \in \text{Comb}[ (\emptyset, \h{1}), (\h{2},\h{3}), \ldots, (\h{2N-2},\h{2N-1} \otimes \h{\text{aux}})]$ for $C$ (blue) is a comb that ``eats'' $C$ and ``spits out'' a quantum state on $\h{2N} \otimes \h{\text{aux}}$. }
		\label{fig:comb_probe}
		\end{figure}

Reference~\cite{Yang2019} introduced \emph{parametrized quantum combs} for parameter estimation with non-Markovian processes. 
Consider an $N$-step non-Markovian process  that consists of $N$ consecutive quantum channels with partially accessible inputs and outputs  concatenated by an inaccessible memory.
Such a process is completely positive (CP) 
%,i.e.,  sending part of any positive semidefinite operator through its input ports would yield a positive semidefinite output. 
and thus, by the Choi-Jamio{l}kowski isomorphism, 
 the non-Markovian process is uniquely characterized by a positive semidefinite Choi operator  \cite{Choi1975}, referred to as its quantum comb \cite{Chiribella2007,Chiribella2008,Chiribella2009}. For a process, the comb is obtained by inserting the second half of the vector $\sum_n|n\>|n\>$ into all input ports (e.g.\ odd-labeled spaces in Fig.~\ref{fig:comb_probe}), where $\{|n\>\}$ is an orthonormal basis of the input ports. Conversely, an operator corresponds to the comb of a process if it satisfies a series of linear constraints (see Definition \ref{def_comb} below).
	We will therefore  denote a non-Markovian process by its comb. In Fig.~\ref{fig:comb_probe}, we can see that a  comb consists of a sequence of $N$ teeth (corresponding to an $N$-step process), each tooth $k=1,\ldots, N$ has an input space $\h{2k-1}$ and an output space $\h{2k}$. We denote by $\text{Comb}[ ( \h{1},\h{2}), (\h{3},\h{4}),  \linebreak \ldots, (\h{2N-1},\h{2N})]$ all combs for a given sequence of input and output ports (characterized, e.g., by odd- and even-labeled Hilbert spaces respectively). We denote by $\mathcal{L}(\spc{H})$ linear operators on $\spc{H}$ and assume all Hilbert spaces to be finite dimensional.
	\begin{defi}[Quantum\ combs\ \cite{Chiribella2007,Chiribella2008,Chiribella2009}] \label{def_comb} A positive semidefinite $C\in\mathcal{L}(\bigotimes_{i=1}^{2N} \h{i})$ is a comb in $\text{Comb}[ ( \h{1},\h{2}), (\h{3},\h{4}), \ldots, (\h{2N-1},\h{2N})]$ if and only if there exist a sequence of positive semidefinite operators $C^{(k)}\in\mathcal{L}(\bigotimes_{i=1}^{2k} \h{i})$ ($k=1, \ldots, N-1$) such that:
		\begin{equation}	
		\label{eq:comb_normalisation}	
		\begin{aligned} 
		& & \tr_{2k}\left[ C^{(k)} \right] &= \mathds{1}_{2k-1} \otimes C^{(k-1)} \quad k=2, \ldots, N-1, \\
		& &  \tr_{2}\left[ C^{(1)} \right] &= \mathds{1}_1\quad \tr_{2N}\left[ C \right] = \mathds{1}_{2N-1} \otimes C^{(N-1)}.  \end{aligned}
		\end{equation}
	\end{defi}
The linear constraints in Eq.~(\ref{eq:comb_normalisation}) are due to normalization and causality. The constraints distinguish between input and output ports (as, for example, it is always the output ports that are traced out) and
ensure that the ports are causally ordered according to the labels. When $N=1$, Eq.~(\ref{eq:comb_normalisation}) reduces to a simple normalization constraint $\tr_{2}\left[ C \right] = \mathds{1}_1$, and $\text{Comb}[ ( \h{1},\h{2})]$ is just the collection of all quantum channels from $\spc{H}_1$ to $\spc{H}_2$.

	\begin{figure}[b]
		\centering
		\includegraphics[width=0.7\columnwidth]{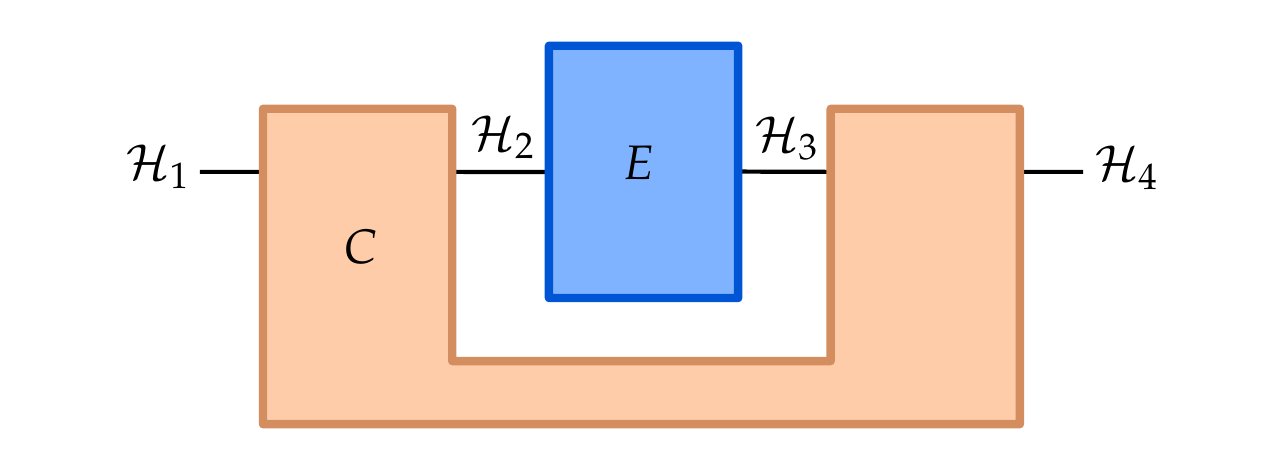}
		\caption{{\bf Link product.} The link product between a comb $C \in \text{Comb}[(\h{1},\h{2}),(\h{3},\h{4})]$ and a quantum channel $E \in \text{Comb}[(\h{2},\h{3})]$ results in a channel $E * F $ from $\h{1}$ to $\h{4}$.}
		\label{fig:link_product}
		\end{figure}

	Two combs can be interlaced with each other following the rule of the link product \cite{Chiribella2008, Chiribella2009}: Given two combs $E \in \mathcal{L}(\bigotimes_{m \in M} \h{m})$ and $F \in \mathcal{L}(\bigotimes_{n \in N} \h{n})$, the link product $E * F \in \mathcal{L}(\bigotimes_{s \in S} \h{s})$ for $S = (M \cup N) \backslash (M \cap N)$, which yields a new comb, is defined by
		\begin{equation}\label{eq:linkprod} E * F := \tr_{M \cap N} \left[ \left( E^{T_{M\cap N}} \otimes \mathds{1}_{N \backslash M}  \right) \left( \mathds{1}_{M \backslash N} \otimes F \right) \right], \end{equation} 
		where $T_{A}$ is the partial transpose performed on Hilbert spaces with indices in $A$.  An example is depicted in Fig.~\ref{fig:link_product}. Also note that in Fig.~\ref{fig:comb_probe} the link product $C\ast T$ yields a quantum state.

	For Markovian quantum metrology, the task is to find an optimal probe state so that, by inserting it to the parametrized channel, the output state has maximal QFI. Similarly, to estimate $\theta$ from a parametrized quantum comb $C_\theta$, the goal is to find an optimal \emph{probe} $T$ consisting of state preparation and control such that $C_\theta\ast T$ is a quantum state with maximal QFI. As depicted in Fig.~\ref{fig:comb_probe}, for $C_\theta\in\text{Comb}[ (\h{1},\h{2}), \dots, (\h{2N-1},\h{2N})]$, the probe should be  $T \in \text{Comb}[ (\emptyset, \h{1}), (\h{2},\h{3}), \linebreak \ldots, (\h{2N-2},\h{2N-1} \otimes \h{\text{aux}})]$, with $\h{\text{aux}}$ being an ancilla and $\emptyset$ denoting a trivial input space. Thus, the probe has input and output spaces complementary to the original quantum comb at the first $(N-1)$ teeth, and $C_\theta * T$ is a quantum state in $\h{2N}\otimes\h{\text{aux}}$.
	It is then natural to define the QFI of a quantum comb as \cite{Yang2019}: 
	\begin{equation} \label{eq:qfi_def}
		J(C_\theta) = \max_{T\in \text{Comb}[ (\emptyset, \h{1}),  \linebreak \ldots, (\h{2N-2},\h{2N-1} \otimes \h{\text{aux}})]} J( C_\theta * T),
	\end{equation}
	where, on the right-hand side, $J(\rho_\theta)$ is the QFI of a quantum state $\rho_\theta$ \cite{Helstrom1969,Holevo2011}. 
	
	The QFI  gives a lower bound on the variance of any unbiased estimator $\hat{\theta}$ via the Cram\'{e}r-Rao bound \cite{Helstrom1969, Holevo2011}, which can be extended to the comb setting \cite{Yang2019}:
		\begin{equation}\label{comb-crb}
		\text{Var}(\hat{\theta}) \geq \frac{1}{\nu J(C_\theta)}
		\end{equation}
	with $\nu$ being the number of times that the experiment is repeated. When $\theta$ is a single parameter, the bound is known to be achievable \cite{Helstrom1969, Holevo2011}.

    The bound (\ref{comb-crb}) establishes the comb QFI $J(C_\theta)$ as the pivotal quantity  that determines the ultimate precision limit. Our main result consists of an analytical formula for $J(C_\theta)$ and an algorithm that computes it using semidefinite programming. We start by  introducing the key notions:
	\begin{defi}[Comb conditional min-entropy \cite{Chiribella2016}]
	 For any $C \in \text{Comb}[(\h{1},\h{2}), \ldots, \linebreak (\h{2N-1},\h{2N})]$, the min-entropy of the $N$th tooth $(\h{2N-1},\h{2N})$ conditioned on the first $(N-1)$ teeth $[N-1] := (\h{1},\h{2}), \ldots, (\h{2N-3},\h{2N-2})$ is defined as
		\begin{equation}  \label{eq:hmincomb}
		\begin{aligned}
		& H_{\text{min}}(N | [N-1])_C := \\
		& - \log_2 \min_{S} \{ \lambda \in \mathbb{R} ~|~ \mathds{1}_{2N-1,2N} \otimes \lambda S \succeq C \},
	\end{aligned}
\end{equation}			
	where $S \in \text{Comb}[(\h{1},\h{2}), \ldots, (\h{2N-3},\h{2N-2})]$.
	\end{defi}
	  Fujiwara and Imai \cite{Fujiwara2008} evaluated the QFI of a quantum state $\rho_\theta$ via \emph{ensemble decompositions}, which are vectors $\{|\phi_{\theta,i}\>\}$ such that $\rho_\theta=\sum_{i}|\phi_{\theta,i}\>\<\phi_{\theta,i}|$. Note that $\{|\phi_{\theta,i}\>\}$ are not required to be orthonormal and thus ensemble decompositions are not unique. Then the QFI of $\rho_\theta$ is determined by the operators $\sum_i|\dot{\phi}_{\theta,i}\>\<\dot{\phi}_{\theta,i}|$,  where the dot stands for the partial derivative with respect to $\theta$. 
	Here (see Ref.~\cite{SM} for details) we extend the notion to combs:
	Since the Choi state of a quantum comb $C_\theta$ is Hermitian, we can find unnormalized vectors $\lket{C_{\theta,i}}$ such that
		\begin{equation} \label{eq:ensemble}
		C_\theta = \sum_{i=1}^q \lket{C_{\theta,i}} \lbra{C_{\theta,i}},
		\end{equation}
	where $q\ge r := \max_\theta\text{rank}(C_{\theta})$. We then define the performance operator of metrology as
		\begin{equation} 
		\Omega_\theta(h) := 4 \sum_{i=1}^q \left(\lket{\dot{\tilde{C}}_{\theta,i}} \lbra{\dot{\tilde{C}}_{\theta,i}}\right)^{T_{1\dots 2N-1}}. 
		\end{equation}
		Here  $\lket{\dot{\tilde{C}}_{\theta,i}}$ is given by
		$\lket{\dot{\tilde{C}}_{\theta,i}}  = \lket{\dot{C}_{\theta,i}} - \im \sum_{j=1}^q h_{ij} \, \lket{C_{\theta,j}}$. The Hermitian matrix $h$ is of dimension $q$ for any $q\ge r$, which captures the nonuniqueness of the ensemble decomposition. One can show that (see below) it is sufficient to consider $q=r$.
	 With these notions we have:
	\begin{theo}
	\label{theorem:qfi} Given a parametrized family of combs $\{C_\theta\}_\theta$ with $C_\theta \in \text{Comb}[(\h{1},\h{2}), \ldots, (\h{2N-1},\h{2N})]$, the QFI of the comb at $\theta$ is given as
		\begin{equation} \label{eq:qfi_result}
			J(C_\theta) =	 d_{2N}  \min_{h \in \text{Herm}(\mathbb{C}^{r})} 2^{- H_{\min}(N | [N-1])_{\Omega_\theta(h)}}
		\end{equation}
	with $d_{2N} := \dim(\h{2N})$ and $\text{Herm}(\mathbb{C}^{r})$ being the set of $r$-dimensional Hermitian matrices.
	\end{theo}
The proof can be found in Ref.~\cite{SM}.
\if\citationorder1
\cite{Sidhu2019,hughston1993complete,Rockafellar1970,Chiribella2012,Watrous2018,Gutoski2012,
Molina2012,Nielsen2000,Ziman2005,Oreshkov2007,Krovi2007,CVX,Altherr2021}
\fi
%The QFI of a comb can be expressed in terms of a conditional min-entropy, which characterises the maximal amount of correlation that can be achieved between the output spaces $\h{\text{aux}}$ and $\h{2N}$ of $\Omega_\theta * T$ where $T$ acts on the first $(N-1)$ teeth of $\Omega_{\theta}$ \cite{Konig2009,Chiribella2016}.  
%The goal of comb metrology is to find a generalized probe, consisting of a suitable input state together with optimal quantum control, whose {\red dual} resembles the performance operator.  
Equation~(\ref{eq:qfi_result}) can be regarded as the \emph{ultimate} formula of QFI, which applies to both Markovian and non-Markovian processes: By setting $C_\theta$ to be the tensor product of $N$ identical Choi operators, Eq.~(\ref{eq:qfi_result}) yields the QFI corresponding to the optimal adaptive strategy of channel estimation (see Ref.~\cite{SM} for details).

In practice, it is  desirable to evaluate the comb QFI numerically. To this purpose, we develop Algorithm \ref{algo:qfi}, which takes as input any parametrized family of combs and outputs the comb QFI. The detailed derivation of Algorithm~\ref{algo:qfi} can be found in Ref.~\cite{SM}.  

	 \begin{algorithm}[h!] 
                \SetKwInOut{Input}{input}\SetKwInOut{Output}{output}

                \Input{ a parametrized family of quantum combs $C_\theta \in \text{Comb}[ (\h{1},\h{2}), \ldots, (\h{2N-1},\h{2N})]$.}
                \Input{the dimension $d_i $ of each $\h{i}$.}
                \Output{the comb QFI $J(C_\theta)$.}
                \BlankLine
                \Begin{
                \tcp{Variables for optimization.}
                	variables $\lambda \geq 0$, $h \in \text{Herm}(\mathbb{C}^r)$, \\ \quad $S^{(k)} \in \text{Herm}(\mathbb{C}^{\prod_{i=1}^{2k} d_i})$ ($1 \leq k \leq N-1$) \;
                	
                	%(\prod_{i=1}^{2k} d_i, \prod_{i=1}^{2k} d_i)$ ($1\le k\le N-1$) Hermitian 
                	
                	find an ensemble decomposition $\{ \lket{C_{\theta,i}} \}_{i=1}^r$ of $C_\theta$ \;
                	\For{$i=1$ \KwTo $r$}{	
                	$\lket{\dot{\tilde{C}}_{\theta,i}} = \lket{\dot{C}_{\theta,i}} - \im \sum_{j=1}^r h_{ij} \lket{C_{\theta,j}}$ \;
                	
                	\For{$m_{2N}=1$ \KwTo $d_{2N}$}{
                	$\lket{c_{i,m_{2N}}} = \sum_{m_1, \ldots, m_{2N-1}} \dot{\overline{\tilde{C}}}_{i,\theta}^{m_1 \dots m_{2N-1} m_{2N}} \lket{m_1 \dots m_{2N-1}}$ \footnote{$\overline{a}$ denotes the complex conjugate of $a$.}\;
                	}
                	}
                	\tcp{Create the performance operator $\Omega_\theta(h)$.}
                	$
	 		A = \left( \begin{array} {ccc|c}
	 		&  & &  \lbra{c_{1,1}} \\
	 		& \mathds{1}_{r\cdot d_{2N}} &  &  	\vdots \\
	 		 &  &  &   \lbra{c_{r,d_{2N}}} \\ \hline
	 		 \lket{c_{1,1}} & 
	 		 \dots & 
	 		 \lket{c_{r,d_{2N}}} & 
	 		\mathds{1}_{2N-1} \otimes \lambda S^{(N-1)} \end{array}\right)$ \;
		 		$\text{minimize}_{\lambda, h, S^{(k)}}$ \qquad $\lambda$ \;%over $h$, $\lambda$, and $\{S^{(k)}\}$ \;
		 		subject to  \quad $ A \succeq 0$,  \qquad $\tr_2\left[ S^{(1)} \right] == \, \mathds{1}_{1}$ \;
 				 \quad $\tr_{2k}\left[ S^{(k)} \right] == \mathds{1}_{2k-1} \otimes S^{(k-1)}$ ($2 \leq k \leq N-1$) \;
		}
		output $J(C_\theta) = 4 \lambda$\;
                \caption{Evaluating the QFI of a comb.}
                \label{algo:qfi}  
         \end{algorithm} 
 
\medskip
\noindent{\it Frequency estimation under non-Markovian noise.} %: quantum phase estimation under non-Markovian noise}
	As an example, we now apply our framework to a specific case of frequency estimation under non-Markovian noise. The task is to estimate an unknown frequency $\omega$ , given $N$ sequential access to a qubit phase gate $e^{-\im Ht}$ with $H=\omega |1\>\<1|$. The whole process, however, is subject to non-Markovian noise, which can be modeled using  a circuit model: 
	At each of the $N$ steps, the system $\h{S}$ collides  with the same environment $\h{E}$ via a unitary interaction $U_{\text{int}}(\tau)$ that lasts for time $\tau$ (see Fig.~\ref{fig:variational_circuit}). The interaction steers information from the system into the environment, potentially reducing the accuracy of  estimation \cite{remark}.
	\if\citationorder1
	\cite{Ciccarello2013,Kretschmer2016}
	\fi
	%\footnote{We remark that in reality the noise and the phase accumulation occur simultaneously rather than sequentially. However, when each step is short enough ($\tau,t\ll1$) the circuit model still yields a good approximation of practical noises (see, e.g., \cite{Ciccarello2013,Kretschmer2016}). }  
	Meanwhile, the accuracy is also influenced by memory effects, as information can also flow back from the environment to the system. Such a task is beyond the existing framework of quantum metrology: One may be tempted to work out the maximal achievable QFI via a brute-force approach, by optimizing over all possible input states and controls. However, even for small systems the brute-force approach may still be unrealistic. For example, adding an ancillary system makes the QFI higher in general, but due to the underlying memory of the environment the size of the ancillary system needed to achieve optimality is hard to determine.
	Theorem \ref{theorem:qfi}, on the other hand, involves optimization on a predetermined system and its efficiency is totally predictable.

	\begin{figure}
	\centering
	\includegraphics[width=\columnwidth]{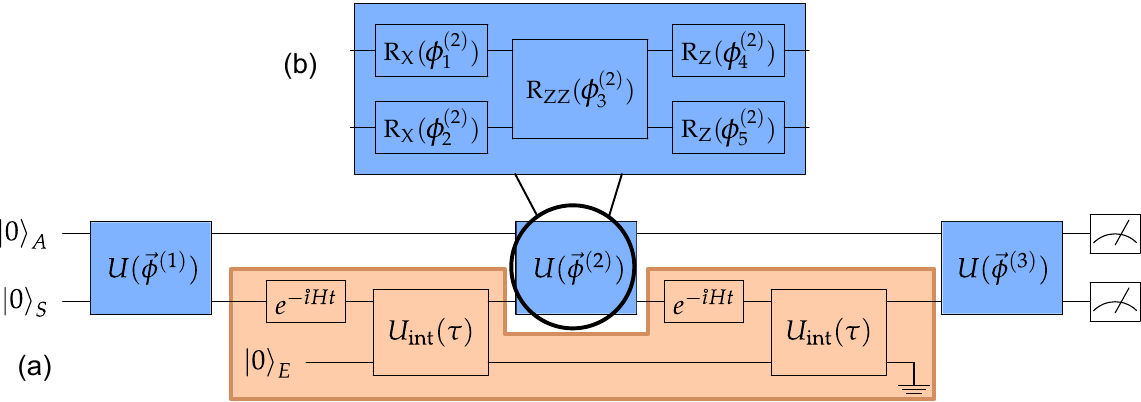}
	\caption{{\bf Noisy non-Markovian frequency estimation using a variational probe.} As shown in (a), the task is to estimate $\omega$ given $N$ accesses ($N=2$ in the figure), each of duration $t$, to a qubit system $S$ with $H=\omega|1\>\<1|$. The system is coupled to a qubit environment $E$ by an interaction $U_{\rm int}(\tau)$, where the interaction time $\tau$ depends on $t$. To estimate $\omega$ we construct a variational probe which consists of two-qubit control operations $U(\vec{\phi}^{(i)})$ for $i=1,2,\dots,N+1$ and projective measurements in the computational basis. In (b) we consider a specific variational probe: The rotations are $R_\sigma(\phi) = e^{-\im \frac{\phi}{2} \sigma}$ for $\sigma = X,Z$ and $R_{ZZ}(\phi) = e^{-\im \frac{\phi}{2} Z \otimes Z}$.  }
	\label{fig:variational_circuit}
	\end{figure}

	Let us now study one concrete case. We set both the system and the environment to be single qubit. For the interaction, we choose $U_{\text{int}}(\tau) = e^{-\im g \tau \,  \text{SWAP}}$ with interaction strength $g$ and  $\text{SWAP}(\ket{v}_S \otimes \ket{w}_E) = \ket{w}_S \otimes \ket{v}_E$ being the swap gate.  We apply Algorithm \ref{algo:qfi} to evaluate the comb QFI. For each round of the evaluation, we fix $N$ and a total time $t_{\text{tot}}$ and set $g=1, \omega = \pi/10$, $t = \tau = t_{\text{tot}} / N$. The environment is initiated to $|0\>$.  

	Plots of the comb QFI for different $N$ and different $t_{\text{tot}}$ can be seen in Fig.~\ref{fig:plots}. We compare between two strategies:  The red, dashed lines stand for the QFI of the ``control-free'' strategy, where the probe consists of an initial state preparation and identity channels between the steps, i.e., only the input state is optimized but no intermediate control is performed between the steps. On the other hand, the QFI of the optimal strategy, where both the initial state preparation and the intermediate control are optimized, equals the comb QFI (red, solid lines) by definition.
	From Fig.~\ref{fig:plots}, it is clear that quantum control between the steps can improve the accuracy, as there is a gap between the comb QFI (red, solid) and the control-free QFI (red, dashed). The gap becomes bigger when $N$ grows larger, as the non-Markovianity increases. For both strategies, the QFI of the noiseless scenario (black, dashed), which is equal to $t_{\rm tot}^2$, is recovered only if the interaction $ U_{\text{int}}(\tau)$ is trivial.

	We can also see the difference between Markovian and non-Markovian noises: In Fig.~\ref{fig:plots}, the blue, solid line corresponds to the maximal achievable QFI of the Markovian setting, where we apply the same interaction between the system and the environment but reset the environment at each step. In the Markovian setting, the whole process can be described by $N$ sequential quantum channels. One can see that the QFIs for both the optimal strategy (blue, solid) and the control-free strategy (blue, dashed) are lower than those for the non-Markovian case, matching a recent finding in Ref.~\cite{Wu2020}. Intuitively, the reason could be that, for the non-Markovian case, the information can be retrieved from the environment using proper control. This is a phenomenon that, without doubt, deserves further investigation.

	We also investigate the same task for other types of non-Markovian noises. These results can be found in Ref.~\cite{SM}.
	
	\begin{figure}[t!]
	\begin{center}
	\includegraphics[width=1\columnwidth]{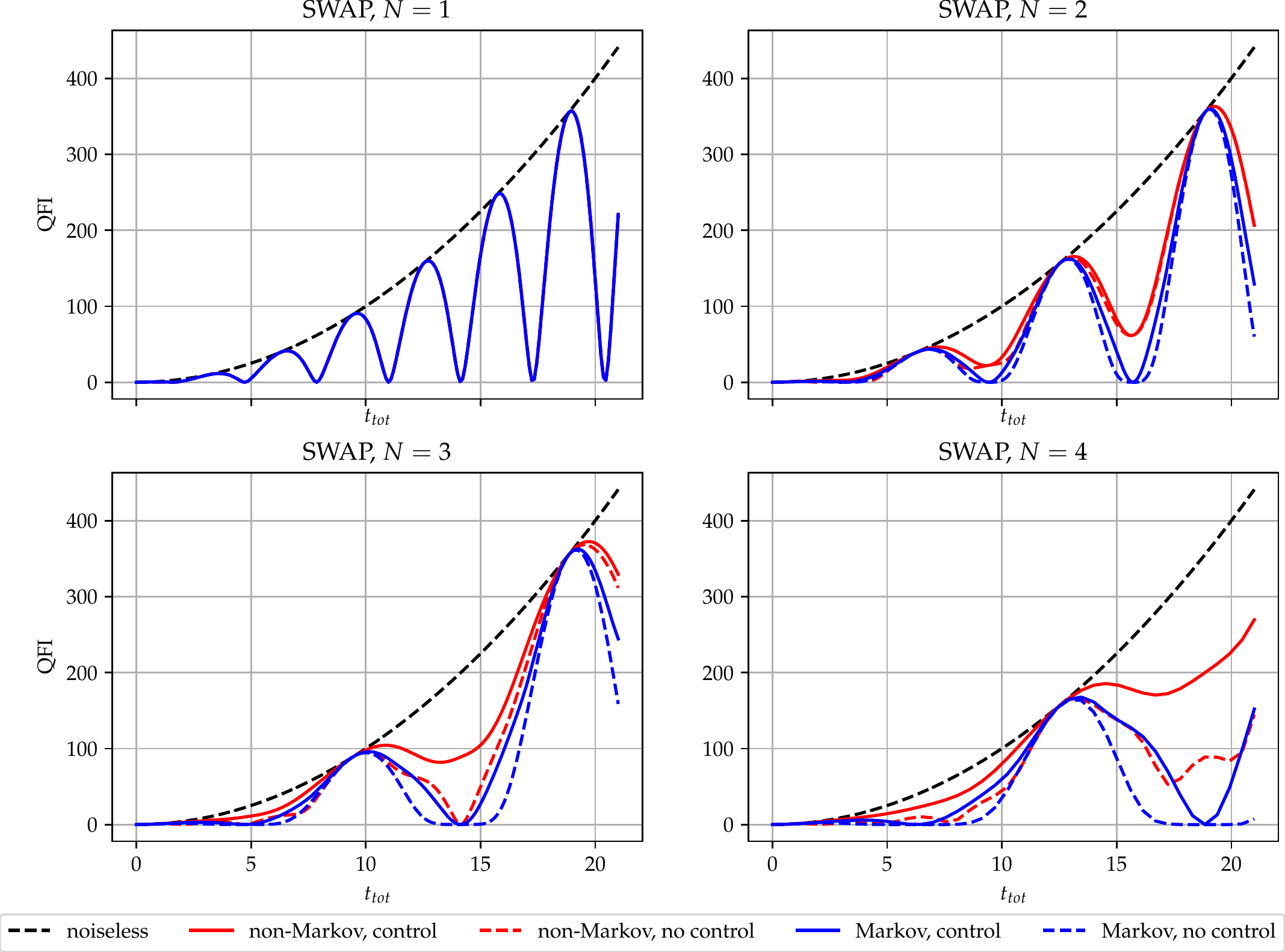}
	\caption{{\bf QFIs as a function of the total sampling time $t_{\text{tot}}$ for noisy frequency estimation.} We consider the comb depicted in Fig.~\ref{fig:variational_circuit} with the SWAP interaction for different value of $N$. As a comparison, we also consider the corresponding Markovian scenario, where the environment state is traced out and reset to $|0\>$ after each interaction. The non-Markovian scenario with optimal control (red, solid) allows for the highest QFI. In particular, the QFI is higher than those of the scenario without control (red, dashed) and its Markovian counterpart (blue, solid). The Markovian scenario without control  (blue, dashed) sees the worst performance. Notice that $N=1$ corresponds to the task of single channel estimation, where all scenarios coincide. When the interaction is trivial ($U_{\rm int}(\tau)= \pm \mathds{1}$), the noiseless scaling (black, dashed) is achieved in all scenarios. }
	\label{fig:plots}
	\end{center}
	\end{figure}

\medskip
\noindent{\it High-performance metrology with variational probes.} Practically, it is meaningful to consider whether the maximal QFI can be achieved by a relatively simple probe. Variational circuits, which promise a range of near-term applications thanks to its relatively simple structure, are ideal candidates. 
Consider the same estimation problem as the previous section.
As shown in Fig.~\ref{fig:variational_circuit} b), we construct a variational probe, which consists of a fixed arrangement (i.e. ansatz) of unitary gates and measurements in the computational basis. Each unitary gate is controlled by a few variables, and the optimal probe is obtained by gradually adjusting these variables to maximise the Fisher information of the output probability distribution. 

Concretely, we use a checkerboard ansatz \cite{Uvarov2020} with $5(N+1)$ parameters in total. From the numerical simulation [Fig.~\ref{fig:variational}], one can see that the performance of the variational probe is close to optimal, implying that it is enough to use it when $N$ is not large. The gap  to the comb QFI increases with $N$, suggesting that it may require more complex circuits to approach optimality for estimating processes with stronger non-Markovianity.

%The data generated by the SDP and the variational circuit can be found in a repository \cite{Altherr2021}.
	
	\begin{figure}[t!]
	\begin{center}
	\includegraphics[width=\columnwidth]{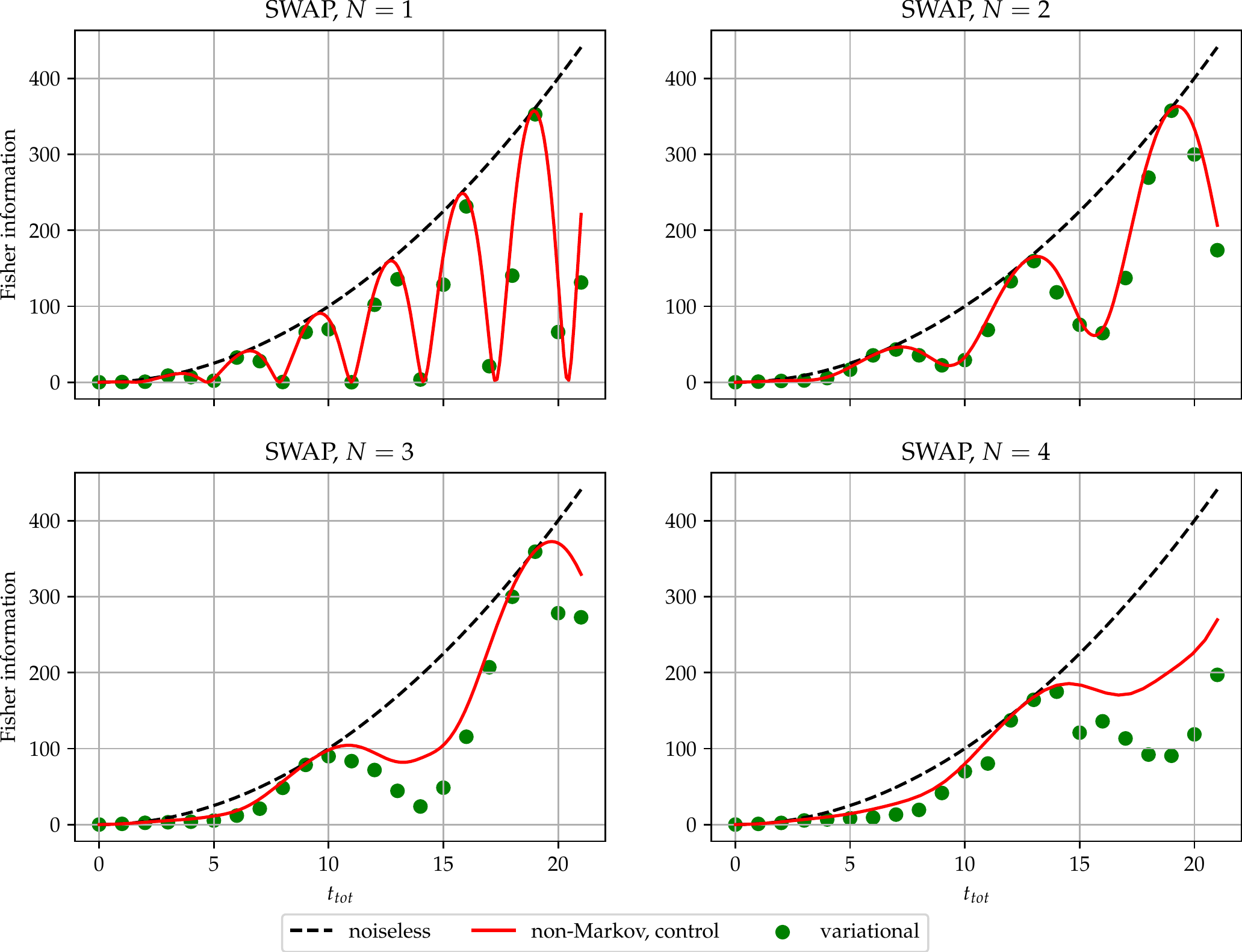}
	\caption{{\bf Performance of the variational probe.} Performance of the variational probe (blue part of Fig.~\ref{fig:variational_circuit}) is evaluated in terms of the Fisher information of the output probability distribution. Here the green dots are the Fisher information of the variational probe's output and the red, solid lines are the maximal attainable Fisher information.}
	\label{fig:variational}
	\end{center}
	\end{figure}

\medskip
\noindent{\em Conclusion.}    We established a fundamental framework of non-Markovian quantum metrology. To this end, we presented a general formula of the QFI of a quantum comb and designed an algorithm which evaluates the exact value of the comb QFI.  
	Our work opens up various directions for further research.
	First, Theorem \ref{theorem:qfi} and Algorithm \ref{algo:qfi} can be readily applied to any concrete task of non-Markovian metrology. 
	%In particular, they would be crucial to the design of testers that probe the degree of non-Markovian noise in deep quantum circuits.
	Second,	it is interesting to explore the asymptotic performance of non-Markovian metrology  using our framework and to capture unique performance limits such as the quantum Zeno limit \cite{Chin2012}. 
	Our method can also be extended to multi-parameter non-Markovian quantum metrology \cite{szczykulska2016multi,Demkowicz2020}, where the tradeoff between the precisions of estimating different parameters \cite{Holevo2011,PhysRevA.61.042312,ragy2016compatibility,Gessner2018,albarelli2019evaluating,yang2019attaining,suzuki2020quantum} plays a major role. 
	Last but not least, we note that an $N$ comb, in general, has degrees of freedom that scales exponentially in $N$. Nevertheless, in many practical scenarios, it is reasonable to assume that the comb and its performance operator have only ${\rm poly}(N)$ significant eigenvalues. It is then promising that an accurate and even more efficient approximate algorithm can be obtained from state-of-the-art methods of semidefinite programming, e.g., Ref.~\cite{yurtsever2021scalable}.

\medskip

We thank Giulio Chiribella, Tobias Sutter, and Sisi Zhou for helpful discussions.
This work is supported by the Swiss National Science Foundation via the National Center for Competence in Research ``QSIT" as well as via project No.\ 200021\_188541 and by the ETH Pauli Center for Theoretical Studies.

\bibliography{Masterthesis}   
\newpage

\begin{widetext}
\appendix

\section{Proof of Theorem 1} 
\label{sec:proof}
	
Here we prove Theorem 1 of the main text by expressing the QFI of a quantum comb as a semidefinite program and deriving its dual form. Note that we use a slightly different formulation of the primal problem that is useful for the numerical implementation. In Subsection \ref{sec:primal}, we express the QFI of a parametrised comb as a semidefinite programm. In Subsection \ref{sec:dualproblemN>1}, we derive the corresponding dual problem for $N > 1$, in Subsection \ref{sec:dual1}, we treat the case $N=1$. In Subsection~\ref{sec:equivalent_decomp}, we handle the minimisation over equivalent decompositions, and in Subsection \ref{sec:dual_minentropy}, we show the changes needed to obtain the dual problem in terms of the conditional min-entropy.
	
	\subsection{The comb QFI as an optimization problem}
	\label{sec:primal}
	In this subsection, we show that the QFI of quantum combs can be expressed as a semidefinite program.  The derivation uses a variational expression for the QFI of quantum states \cite[Theorem 1]{Fujiwara2008}, which we introduce here as the following lemma.
	
	\begin{lemma} \label{lem:fujiwara}
Let $\{\rho_\theta ~|~ \theta \in \Theta\}$ be a smooth curve of quantum states of constant rank $r$ and $q \geq r$ an arbitrary integer. The QFI of a state $\rho_{\theta}$ is
		\begin{align} \label{eq:qfi_state} J(\rho_\theta) = 4 \min_{ \{\ket{\psi_{\theta,i}} \} }
		\sum_{i=1}^q \tr\left[  \lket{\dot{\psi}_{\theta,i}} \lbra{ \dot{\psi}_{\theta,i}} \right]. \end{align}
		Here, the minimisation is over all $\{\ket{\psi_{\theta,i}}\}$  such that
		\begin{equation} \label{eq:decomposition}
		\rho_\theta = \sum_{i=1}^q \ket{\psi_{\theta,i}} \bra{\psi_{\theta,i}}.
		\end{equation}
\end{lemma}

	Now, consider a parametrised comb $C_\theta \in \text{Comb}[(\h{1},\h{2}), \ldots, (\h{2N-1}, \h{2N})]$ with $\theta \in \Theta$ and its QFI [Eq.~(3) in the main text]
	\begin{equation} \label{eq:qfi_appendix}
		J(C_\theta) = \max_{T} J( C_\theta * T)
	\end{equation}
	for $T \in \text{Comb}[(\emptyset, \h{1}),\ldots, (\h{2N-2},\h{2N-1}\otimes \h{\text{aux}})]$. We split the derivation into multiple lemmata.
	
	\begin{lemma}
	The QFI of a comb $C_\theta$ can be expressed as
	\begin{equation}
		J(C_{\theta}) = \max_{T} \min_{\{|C_{\theta,i}\> \}} \tr[ T \, (\mathds{1}_{\text{aux}} \otimes \tr_{2N}[\Omega_\theta])].
		\end{equation}
	\end{lemma}	
	
	\begin{proof} We first fix a probe $T$
	 and consider a decomposition
	\begin{equation} \label{eq:ensemble_appendix}
		C_\theta = \sum_{i=1}^q \lket{C_{\theta,i}} \lbra{C_{\theta,i}}
		\end{equation}	
	with $q \geq \max \{ \text{rank}(C_\theta) ~|~ \theta \in \Theta \}$ components. 	
	Due to the linearity of the link product and the convexity of the QFI \cite{Sidhu2019}, the maximal QFI will be achieved for pure $T = \ket{T} \bra{T}$.\footnote{For general $T = \sum_l p_l \ket{T_l} \bra{T_l}$, Eq.~\eqref{eq:link_decomp} reads $C_\theta * T = \sum_l p_l \sum_i \ket{v_{\theta,il}}\bra{v_{\theta,il}}$ for an appropriate $\ket{v_{\theta,il}}$ (compare Eq.~\eqref{eq:link_vil}). Convexity of the QFI implies $J(C_\theta * T) \leq \sum_l p_l J( \sum_i \lket{v_{\theta,il}} \lbra{v_{\theta,il}}) = \sum_l p_l J( C_\theta * \ket{T_l}\bra{T_l})$. The maximal QFI is thus achieved for pure $T$. } Using the definition of the link product [Eq.~(2) of the main text], we can write	
		\begin{equation} \label{eq:link_decomp}
		C_\theta * T = \sum_{i=1}^q\ket{v_{\theta,i}} \bra{v_{\theta,i}},
		\end{equation}
		where $\ket{v_{\theta,i}}$ can be written in components of $\ket{C_{\theta,i}}$ and $\ket{T}$ as
		\begin{equation} \label{eq:link_vil}
		\ket{v_{\theta,i}} = \sum_{m_1, \ldots, m_{2N}, m_{\text{aux}}} C_{\theta,i}^{m_1 \dots m_{2N-1} m_{2N}} T^{m_1 \dots m_{2N-1} m_{\text{aux}}} \ket{m_{2N} m_{\text{aux}}}.
		\end{equation}
		
		Now, there arise some subtle issues regarding the existence of the QFI of a quantum comb that trace back to the existence of the QFI of a quantum state (Lemma~\ref{lem:fujiwara}). In particular, we assume:
	\begin{enumerate}
	\item The comb $C_\theta$ has a decomposition \eqref{eq:ensemble_appendix} with number of components equal to $r:=\max\{ \text{rank}(C_\theta) ~|~ \theta \in \Theta \}$ such that each component $\ket{C_{\theta,i}}$ is continuously differentiable.
	\item We assume that for a given probe $T$ the link product $C_\theta * T$ has constant rank for all $\theta \in \Theta$.
\end{enumerate}		
	Under these conditions, for each probe $T$ the set $\{C_\theta * T ~|~ \theta \in \Theta \}$ is a smooth curve of quantum states with constant rank that immediately allows for application of Lemma~\ref{lem:fujiwara}. Applying Eq.~\eqref{eq:qfi_state} to Eq.~\eqref{eq:link_decomp}, we find the QFI of the output to be
		\begin{equation} \label{eq:qfi_link}
		J(C_{\theta} * T) = \min_{\{|C_{\theta,i}\> \}} \tr[ T \, (\mathds{1}_{\text{aux}} \otimes \tr_{2N}[\Omega_\theta])],
		\end{equation}
	the minimisation is taken over all decompositions as in Eq.~\eqref{eq:ensemble_appendix} and we use the performance operator 
		\begin{equation}
		\Omega_\theta = 4 \sum_{i=1}^q \left(\lket{\dot{C}_{\theta,i}} \lbra{\dot{C}_{\theta,i}} \right)^{T_{1\dots 2N-1}}. 
		\end{equation}
		Including the maximization over all probes, we obtain the desired result.
		
	\end{proof}

	We now show how to deal with the minimization over equivalent decompositions.
	
	\begin{lemma}
	The QFI of a comb $C_\theta$ can be expressed as 
	\begin{equation} \label{eq:qfi_expr1}
		J(C_{\theta}) = \max_{T} \min_{h}  \tr[ T \, (\mathds{1}_{\text{aux}} \otimes \tr_{2N}[\Omega_\theta(h)])].
		\end{equation}	
	\end{lemma}	
	
	\begin{proof}
	 Let us fix an ensemble decomposition $\{ \ket{C_{\theta,i}}\}$ of size $r = \max\{ \text{rank}(C_\theta) ~|~ \theta \in \Theta \}$. We represent $\ket{C_{\theta,i}}$ as columns vectors of $A_\theta^{(0)} := [ \ket{C_{\theta,1}}, \ldots, \ket{C_{\theta,r}}]$ and extend it to an ensemble decomposition of size $q\geq r$ by considering $A_\theta := A_\theta^{(0)} I_{r,q}$ where $I_{r,q} = [ \mathds{1}_r ~|~ 0_{r,q-r}]$ and $0_{r,q-r}$ denotes the $r \times (q-r)$ zero matrix. Arbitrary ensemble decompositions with $q$ components can be related to $A_\theta$ via a unitary $V_\theta$: $\tilde{A}_\theta = A_\theta V_\theta$ \footnote{For arbitrary ensemble decompositions $\sum_{i=1}^q\ket{A_i}\bra{A_i}$ and $\sum_{i=1}^q\ket{B_i}\bra{B_i}$ of a density matrix $\rho$, consider the corresponding bipartite states $\ket{\psi^{A/B}}:=\sum_{i=1}^q\ket{A_i/B_i}\ket{i}$, where $\{\ket{i}\}$ is an orthonormal basis. Since $\ket{\psi^A}\bra{\psi^A}$ and $\ket{\psi^B}\bra{\psi^B}$ are two purifications of the same density matrix, there exists a unitary $V$ such that $(I\otimes V)\ket{\psi^A}=\ket{\psi^B}$ (known as the HJW purification theorem \cite{hughston1993complete}) and we get the desired relation between the two ensemble decompositions. The argument can be easily extended to $\rho$ with non-unit trace.}. The performance operator can be written as
		$\Omega_\theta = 4 (\dot{\tilde{A}}_\theta \dot{\tilde{A}}_{\theta}^\dagger)^{T_{1 \ldots 2N-1}}$ and the expression in brackets evaluates to
	\begin{align}
	\dot{\tilde{A}}_{\theta} \dot{\tilde{A}}_{\theta}^\dagger = (\dot{A}_{\theta} V_\theta + A_\theta \dot{V}_{\theta})( V_\theta^\dagger \dot{A}_{\theta}^\dagger + \dot{V}_{\theta}^\dagger A_{\theta}^\dagger) =
	\dot{A}_{\theta} \dot{A}_{\theta}^\dagger + 
	A_{\theta} \dot{V}_{\theta} V_{\theta}^\dagger \dot{A}_{\theta}^\dagger - \dot{A}_{\theta} \dot{V}_{\theta} V_{\theta}^\dagger A_\theta^\dagger - A_\theta \dot{V}_{\theta} V_\theta^\dagger \dot{V}_{\theta} V_\theta^\dagger A_\theta^\dagger,
	\end{align}		
	where we use $V_\theta^\dagger V_\theta = \mathds{1} = V_\theta V_\theta^\dagger$ and $V_{\theta} \dot{V}_\theta^\dagger = - \dot{V}_\theta V_\theta^\dagger$.
	 Notice in the above expression that the dependency on $V_\theta$ enters only via the Hermitian matrix $h :=\im\dot{V}_\theta V_\theta^\dagger$. Equivalently we can write
	\begin{align} \label{eq:aadagger}
	\dot{\tilde{A}}_{\theta} \dot{\tilde{A}}_{\theta}^\dagger = 
	(\dot{A}_{\theta} - \im A_\theta h)(\dot{A}_{\theta} - \im A_\theta h)^\dagger .
	\end{align}	 	
	It is thus sufficient to consider derivatives of the form	
 		\begin{align} \label{eq:ntilde} 
 		\lket{\dot{\tilde{C}}_{\theta,j}} &= \lket{\dot{C}_{\theta,j}} - \im \sum_{j=1}^q h_{ji} \, \lket{C_{\theta,i}}
 		\end{align}
 		
 	and we write
 		\begin{align}
 		\Omega_\theta(h)&= 4\sum_{j=1}^q \left(|\dot{\tilde{C}}_{\theta,j}\>\<\dot{\tilde{C}}_{\theta,j}|\right)^{T_{1 \dots 2N-1}}
 		\end{align}
 		for the corresponding performance operator.
 		The QFI can then be expressed as
 		\begin{equation} \label{eq:qfi_expr1}
		J(C_{\theta}) = \max_{T} \min_{h}  \tr[ T \, (\mathds{1}_{\text{aux}} \otimes \tr_{2N}[\Omega_\theta(h)])].
		\end{equation}	
		\end{proof}

	Next, we show that it is sufficient to consider Hermitian matrices of size $r \times r$.	
		
	\begin{lemma} The minimum over $h$ in Eq.~\eqref{eq:qfi_expr1} is achieved for a $r \times r$ Hermitian.
	\end{lemma}
	
	\begin{proof} We insert $A_\theta = A_\theta^{(0)} I_{r,q}$ into Eq.~\eqref{eq:aadagger} and find	
	\begin{align}
	\dot{\tilde{A}}_{\theta} \dot{\tilde{A}}_{\theta}^\dagger = 
	\dot{A}_{\theta}^{(0)} I_{r,q} I_{r,q}^\dagger (\dot{A}_{\theta}^{(0)})^\dagger  
	-\im A_{\theta}^{(0)} I_{r,q} h I_{r,q}^\dagger (\dot{A}_{\theta}^{(0)})^\dagger + \im \dot{A}_{\theta}^{(0)} I_{r,q} h I_{r,q}^\dagger (A_\theta^{(0)})^\dagger + A_\theta^{(0)} I_{r,q} h^2 I_{r,q}  (A_\theta^{(0)})^\dagger,
	\end{align}	
	We decompose 
		\begin{align}
		h = \begin{pmatrix} h_r & h_{2} \\ h_{2}^\dagger & h_{q-r} \end{pmatrix}
		\end{align}
	into a $r \times r$ Hermitian $h_r$, a $(q-r) \times (q-r)$ Hermitian $h_{q-r}$ and a $r \times (q-r)$ matrix $h_{2}$. We note that $I_{r,q} I_{r,q}^\dagger = \mathds{1}_r$, $I_{r,q} h I_{r,q}^\dagger = h_{r}$, $I_{r,q} h^2 I_{r,q}^\dagger = h_r^2 + h_{2} h_{2}^\dagger$ and thus
	\begin{align}
	\dot{\tilde{A}}_{\theta} \dot{\tilde{A}}_{\theta}^\dagger =& 
	\dot{A}_{\theta}^{(0)} (\dot{A}_{\theta}^{(0)})^\dagger  
	-\im A_{\theta}^{(0)} h_r (\dot{A}_{\theta}^{(0)})^\dagger + \im \dot{A}_{\theta}^{(0)}h_r (A_\theta^{(0)})^\dagger + A_\theta^{(0)} (h_r^2 + h_2 h_2^\dagger)  (A_\theta^{(0)})^\dagger \\
	=& (\dot{A}_{\theta}^{(0)} - \im A_{\theta}^{(0)} h_r) (\dot{A}_{\theta}^{(0)} - \im A_{\theta}^{(0)} h_r)^\dagger + A_\theta^{(0)} h_2 h_2^\dagger (A_\theta^{(0)})^\dagger.
	\end{align}
	The last term $A_\theta^{(0)} h_2 h_2^\dagger (A_\theta^{(0)})^\dagger$ is positive semidefinite. If we denote by $\Omega_\theta(h)$ the performance operator that correponds to the ensemble decomposition $\tilde{A}_\theta$ and by $\Omega_\theta(h_r) := 4 (\dot{A}_{\theta}^{(0)} - \im A_{\theta}^{(0)} h_r) (\dot{A}_{\theta}^{(0)} - \im A_{\theta}^{(0)} h_r)^\dagger$ the performance operator of an ensemble decomposition of size $r$, we note that
	\begin{align}
	\Omega_\theta(h) \succeq \Omega_\theta(h_r).
	\end{align}
	Since Eq.~\eqref{eq:qfi_expr1} is linear in $\Omega_\theta(h)$, it follows that the minimum over equivalent decompositions is achieved for a decomposition of size $r$.
	\end{proof}

	In order to simplify the problem, we are going to exchange the minimisation and maximisation. We are allowed to do so since the function we are interested in is convex in $h$ and concave in $T$ \cite{Rockafellar1970}:
	\begin{equation}
		J(C_{\theta}) = \min_{h} \max_{T} \tr[ T \, (\mathds{1}_{\text{aux}} \otimes \tr_{2N}[\Omega_\theta])].
		\end{equation}

	%Furthermore, the sets over which we extremise are bounded due to the trace constraints on combs. According to \cite[corollary 37.3.2]{Rockafellar1970} this allows us to exchange the minimisation and maximisation. 
	%Since the problem is linear in $T$, we can formulate it as a semidefinite program. 
 We first fix $h$ and maximise over all probes. The problem can be stated as
	\begin{equation} \label{eq:primal}
		\begin{aligned}  
		& \text{maximise} & & \tr \left[ T \,   (\mathds{1}_{\text{aux}} \otimes \tr_{2N}[\Omega_\theta(h)]) \right], \\	
%		& \text{subject to} & &  T \in \text{Comb}[ (\emptyset,\h{1}), \ldots, (\h{2N-2},\h{2N-1}\otimes \h{\text{aux}})].		
		& & \tr_{2N-1, \text{aux}}\left[ T \right] &= \mathds{1}_{2N-2} \otimes T^{(N-1)}, \\
		&  & \tr_{2k-1}\left[ T^{(k)} \right] &= \mathds{1}_{2k-2} \otimes T^{(k-1)}, \quad k=2,\ldots, N-1, \\
		&   &\tr_{1}\left[ T^{(1)} \right] &= 1, 
		\end{aligned}
		\end{equation}	 
	where $T$, $T^{(k)}$ for $k=1,\ldots, N-1$  are positive semidefinite and the equality constraints arise from the definition of the comb [Eq.~(1) of the main text]. 
Since the objective function and equality constraints are linear in $T$, this is a semidefinite program, which has efficient numerical methods. In addition, we can convert the comb QFI into a minimization problem using (strong) duality.
	
	\subsection{Dual problem for $N>1$}
	\label{sec:dualproblemN>1}
	In this section, we are going to derive the dual problem of the primal problem Eq.~\eqref{eq:primal} for $N> 1$. 
	
	\begin{lemma} The dual problem of the primal problem Eq.~\eqref{eq:primal} for $N>1$ is given as
		\begin{equation} \label{eq:dual_appendix}
	\begin{aligned}
	& \text{minimize} & & S^{(0)}, \\
	& \text{subject to} &  \mathds{1}_{1} \, S^{(0)} &\succeq \tr_{2}\left[ S^{(1)} \right], \\
	& & \mathds{1}_{2k-1} \otimes S^{(k-1)} &\succeq \tr_{2k}\left[ S^{(k)} \right], \quad k=2,\ldots, N-1 \\
	& & \mathds{1}_{2N-1} \otimes S^{(N-1)} &\succeq \tr_{2N}[\Omega_\theta(h)].
	\end{aligned}
	\end{equation}
	In addition, strong duality holds, that is the solution of the dual problem coincides with the solution of the primal problem.	
	\end{lemma}
	
	\begin{proof}
	 We follow the approach of Ref.~\cite{Chiribella2012} and cast the primal problem  in the standard form:
	 \begin{equation} \label{eq:primal_standard}
	\begin{aligned} 
	& \text{maximize} \quad & \tr \left[\mathfrak{T} \mathfrak{C} \right] & \\
	& \text{subject to} \quad &\mathrm{L}(\mathfrak{T}) = \mathfrak{O} & \\
	&	\quad & \mathfrak{T} \succeq 0 & \end{aligned}
	\end{equation}
	where we compose the following quantities:
	\begin{equation}
	\begin{aligned}
	& \mathfrak{T} =  \left(\bigoplus_{k=1}^{N-1} T^{(k)}\right) \oplus T, \quad \mathrm{L}(\mathfrak{T}) = \bigoplus_{k=1}^{N} L^{(k)},\\
	& \mathfrak{C} = \left(\bigoplus_{k=1}^{N-1} 0^{(k)}\right) \oplus \left( \mathds{1}_{\text{aux}} \otimes  \tr_{2N}[\Omega_\theta(h)] \right), \\
	& \mathfrak{O} =  1 \oplus \left(\bigoplus_{k=2}^{N} 0^{(k)} \right).
	\end{aligned}
	\end{equation}
	
	 Here, $0^{(k)}$ denotes the zero matrix on an appropriately chosen space. By composing these block matrices, we can transform the constraints on the probe in one single constraint.
	 %$\bigotimes_{i=1}^{2k-1} \h{i} \otimes \h{2N+1}$ 
	 Furthermore, we decompose
	 	\begin{equation} \label{eq:l}
		\begin{aligned}
		& L^{(1)} = \tr_{1}\left[ T^{(1)} \right] \\
		& L^{(k)} = \tr_{2k-1}\left[ T^{(k)} \right] - \mathds{1}_{2k-2} \otimes T^{(k-1)}, \quad \text{for } k=2,\ldots, N-1, \\
		& L^{(N)} = \tr_{2N-1,\text{aux}}\left[ T \right] - \mathds{1}_{2N-2} \otimes T^{(N-1)}.
		\end{aligned}
		\end{equation}
		
	The dual problem of the problem Eq.~\eqref{eq:primal_standard} is given as \cite{Watrous2018} 
	\begin{equation} \label{eq:dual_standard} 
	\begin{aligned}
	& \text{minimize} \quad & \tr \left[ \mathfrak{S} \mathfrak{O} \right] \\
	& \text{subject to} \quad & \mathrm{L}^\dagger(\mathfrak{S}) \succeq \mathfrak{C}, \\
	& & \mathfrak{S} \, \text{Hermitian},
	\end{aligned}
	\end{equation}
	where we decompose
	\begin{align}
	 \mathfrak{S} = \bigoplus_{k=0}^{N-1} S^{(k)}, \quad \mathrm{L}^\dagger(\mathfrak{S}) = \bigoplus_{k=1}^{N} (L^\dagger)^{(k)},
	\end{align}
	and $\mathrm{L}^\dagger$ is the dual map with respect to the Hilbert-Schmidt product: It satisfies $\tr\left[ \mathfrak{S} \mathrm{L}(\mathfrak{T})\right] = \tr\left[ \mathrm{L}^\dagger(\mathfrak{S}) \mathfrak{T} \right]$ for all $\mathfrak{S}, \mathfrak{T}$. 
	Using this definition, we find for $N>1$
		\begin{equation} \label{eq:ldagger}
		\begin{aligned}
		& & (L^\dagger)^{(1)} &=  \mathds{1}_{1} S^{(0)} - \tr_2\left[ S^{(1)} \right], \\
		& & (L^\dagger)^{(k)}  &= \mathds{1}_{2k-1} \otimes S^{(k-1)} - \tr_{2k}\left[ S^{(k)}\right] & \quad \text{for } k=2, \ldots, N-1, \\
		& & (L^\dagger)^{(N)}  &= \mathds{1}_{2N-1,\text{aux}} \otimes S^{(N-1)}.
		\end{aligned}
		\end{equation}

	Inserting these decompositions into Eq.~\eqref{eq:dual_standard}, we obtain
	\begin{equation}
	\begin{aligned}
	& \text{minimize} & & S^{(0)}, \\
	& \text{subject to} &  \mathds{1}_{1} \, S^{(0)} &\succeq \tr_{2}\left[ S^{(1)} \right], \\
	& & \mathds{1}_{2k-1} \otimes S^{(k-1)} &\succeq \tr_{2k}\left[ S^{(k)} \right], \quad k=2,\ldots, N-1 \\
	& & \mathds{1}_{2N-1,\text{aux}} \otimes S^{(N-1)} &\succeq \mathds{1}_{\text{aux}} \otimes  \tr_{2N}[\Omega_\theta(h)].
	\end{aligned}
	\end{equation}
		The last constraint is equivalent to $\mathds{1}_{2N-1} \otimes S^{(N-1)} \succeq  \tr_{2N}[\Omega_\theta(h)]$.
				
	Note that $S^{(N-1)} \succeq 0$ since $\Omega_\theta(h) \succeq 0$. This implies recursively that $S^{(k)} \succeq 0$ for all $k=0,1,\ldots,N-1$. In addition, there exists an $\mathfrak{S}$ with $\mathrm{L}^\dagger(\mathfrak{S}) \succ \mathfrak{C}$. For example, we can set recursively
	\begin{equation} \begin{aligned}
	& & S^{(N-1)} &= 2 \norm{\tr_{2N}[\Omega_\theta(h)]}_{\infty} \, \mathds{1}_{1,2,\dots, 2N-2}, \\
	& & S^{(k)} &= 2 \tr_{2k+1,2k+2} \left[ S^{(k+1)} \right], \quad k=0, \ldots, N-2. 
	\end{aligned} \end{equation}
	This implies, together with the fact that the QFI is upper-bounded \cite[Theorem 1]{Yang2019}, that the assumptions of Slater's theorem on strong duality \cite{Gutoski2012, Molina2012} are satisfied.
	\end{proof}	
	
	 The following lemma shows that we can turn the inequality constraints into equality constraints.
	
	\begin{lemma} For every $\mathfrak{S}$ that satisfies the bounds in the dual problem Eq.~\eqref{eq:dual_appendix}, we can construct a $\tilde{\mathfrak{S}} = \bigoplus_{k=0}^{N-1} \tilde{S}^{(k)}$ such that the first $(N-1)$ inequalities become equalities and the dual form reads
	\begin{equation} \begin{aligned}
	& \text{minimize} & & \tilde{S}^{(0)}, \\
	& \text{subject to}  & \mathds{1}_{1} \, \tilde{S}^{(0)} &= \tr_{2}\left[ \tilde{S}^{(1)} \right], \\
	& & \mathds{1}_{2k-1} \otimes \tilde{S}^{(k-1)} &= \tr_{2k}\left[  \tilde{S}^{(k)} \right], \quad k=2,\ldots, N-1 \\
	& & \mathds{1}_{2N-1} \otimes \tilde{S}^{(N-1)} & \succeq  \tr_{2N}[\Omega_\theta(h)]
	\end{aligned} \end{equation}	
	\end{lemma}

	\begin{proof} We prove the lemma by induction. Assume that $\mathfrak{S}$ satisfies the constraints Eq.~\eqref{eq:dual_appendix}. Set 
	\begin{equation}  \label{eq:transformIntoEquality}
		\begin{aligned} 
		& & \tilde{S}^{(0)} &= S^{(0)}, \\
		& & \delta^{(1)} &= \mathds{1}_{1} \tilde{S}^{(0)} - \tr_2\left[S^{(1)}\right] \succeq 0, \\
		& & \tilde{S}^{(1)} &= S^{(1)} + \delta^{(1)} \otimes \rho_2, \\
		& & \tilde{S}^{(j)} &= S^{(j)}, \quad j=2,\ldots,N-1, \end{aligned} \end{equation}
	where $\rho_2$ is an arbitrary quantum state in $\mathcal{S}(\h{2})$. It follows that $\tr_2[ \tilde{S}^{(1)}] = \mathds{1}_{1}\tilde{S}^{(0)}$, hence $\tilde{\mathfrak{S}}$ achieves equality in the first constraint. Moreover, since $\delta^{(1)} \succeq 0$, we have $\mathds{1}_3 \otimes \tilde{S}^{(1)} \succeq \mathds{1}_3 \otimes S^{(1)}$ and since $S^{(1)}$ satisfies the second constraint, it holds that $\mathds{1}_3 \otimes S^{(1)} \succeq \tr_4[ S^{(2)}] = \tr_4[ \tilde{S}^{(2)}]$. Hence, $\mathds{1}_3 \otimes \tilde{S}^{(1)} \succeq \tr_4 [ \tilde{S}^{(2)} ]$ and $\tilde{S}^{(1)}$ satisfies the second constraint. In total, $\tilde{\mathfrak{S}}$ has the same objective value as $\mathfrak{S}$, satisfies all constraints and achieves equality in the first constraint.
	
	Now, assume that $\mathfrak{S}$ achieves equality in the first $1 \leq k \leq N-2$ constraints and define
		\begin{equation} \begin{aligned}
		& & \tilde{S}^{(j)} &= S^{(j)}, \quad j=1,\ldots, k, \\
		& & \delta^{(k+1)} &= \mathds{1}_{2k+1} \otimes \tilde{S}^{(k)} - \tr_{2k+2}\left[ S^{(k+1)} \right] \succeq 0, \\
		& &	\tilde{S}^{(k+1)} &= S^{(k+1)} + \delta^{(k+1)} \otimes \rho_{2k}, \\
		& &	\tilde{S}^{(j)} &= S^{(j)}, \quad j=k+2, \ldots, N-1, \end{aligned} \end{equation}
	where $\rho_{2k}$ is an arbitrary quantum state on $\mathcal{S}(\h{2k})$. Similar to before, $\tilde{\mathfrak{S}}$ achieves the same objective value as $\mathfrak{S}$, satisfies all constraints and achieves equality in the first $(k+1)$ constraints. By induction, it follows that for every $\mathfrak{S}$ satisfying the constraints, there exists $\tilde{\mathfrak{S}}$ achieving the same objective value and satisfying the first $(N-1)$ constraints with equality.
	\end{proof}
	
	In summary, we have found the dual problem
	\begin{equation} \begin{aligned}
	& \text{minimize} & & S^{(0)}, \\
	& \text{subject to}  & \mathds{1}_{1} \, S^{(0)} &= \tr_{2}\left[ S^{(1)} \right], \\
	& & \mathds{1}_{2k-1} \otimes S^{(k-1)} &= \tr_{2k}\left[ S^{(k)} \right], \quad k=2,\ldots, N-1, \\
	& & \mathds{1}_{2N-1} \otimes S^{(N-1)} & \succeq  \tr_{2N}[\Omega_\theta(h)]
	\end{aligned} \end{equation}
	and shown that its solution is equal to the QFI of the comb $C_\theta$. Note that the first two equality constraints imply that $ S^{(N-1)} / S^{(0)}$ is a quantum comb in $\text{Comb}[(\h{1},\h{2}), \ldots, (\h{2N-3},\h{2N-2})]$.
	Thus, by setting $\lambda = S^{(0)}$ and $S_{\text{nor}}^{(k)} = \frac{S^{(k)}}{\lambda}$ for $k=1,\ldots,N-1$, the dual problem corresponds to:
	\begin{equation} \label{eq:dual_normalised} \begin{aligned} 
	& \text{minimise} & & \lambda \\
	& \text{subject to} & \mathds{1}_{1} &= \tr_{2}\left[ S_{\text{nor}}^{(1)} \right], \\
	& & \mathds{1}_{2k-1} \otimes S_{\text{nor}}^{(k-1)} &= \tr_{2k}\left[ S_{\text{nor}}^{(k)} \right], \quad k=2, \ldots, N-1, \\
	& & \mathds{1}_{2N-1} \otimes \lambda S_{\text{nor}}^{(N-1)} & \succeq  \tr_{2N}[\Omega_\theta(h)].
	\end{aligned} \end{equation}
	By redefining $S^{(k)} := S^{(k)}_{\text{nor}}$ we recover the problem in Algorithm~1 without minimising over equivalent decompositions.

	 \subsection{Dual problem for $N=1$} 
	\label{sec:dual1}	 
	 In this subsection, we derive the dual problem for the case $N=1$, i.e., when the comb is a quantum channel. 
	
	\begin{lemma} The comb QFI for $N=1$ is given by
		\begin{equation}
		J(C_\theta) = \min_{h} \norm{ \tr_2[ \Omega_\theta(h) ]}_{\infty}.
		\end{equation}
	This agrees with a result of Ref.~\cite[Theorem 4]{Fujiwara2008}.
	\end{lemma}
	
	\begin{proof}	 
	 First we start with the primal problem Eq.~\eqref{eq:primal} and cast it into the standard form Eq.~\eqref{eq:primal_standard} 
	\begin{equation} \begin{aligned}
	& \text{maximize} \quad & \tr \left[\mathfrak{T} \mathfrak{C} \right], & \\
	& \text{subject to} \quad &\mathrm{L}(\mathfrak{T}) = \mathfrak{O}, & \\
	&	\quad & \mathfrak{T} \succeq 0, & \end{aligned} \end{equation}
	where
	 	\begin{equation} \mathfrak{T} = T, \quad \mathfrak{C} = \mathds{1}_{\text{aux}} \otimes 
	 	 \tr_{2}[\Omega_\theta(h)], \quad
	 	\mathrm{L}(\mathfrak{T}) = \tr_{\text{aux},1}\left[ T \right], \quad
	 	\mathfrak{O} = 1. \end{equation}
	These definitions coincide with the ones made for $N > 1$. The dual problem takes the form 
	\begin{equation} \begin{aligned}
	& \text{minimize} \quad & \tr \left[ \mathfrak{S} \mathfrak{O}, \right] \\
	& \text{subject to} \quad & \mathrm{L}^\dagger(\mathfrak{S}) \succeq \mathfrak{C},
	\end{aligned} \end{equation}
	where we compose $\mathfrak{S} = S^{(0)},
	 \mathrm{L}^\dagger(\mathfrak{S}) = (L^\dagger)^{(1)}$.
	Now 
	\begin{equation} 
	\begin{aligned}
	& & \tr\left[ \mathrm{L}(\mathfrak{T}) \mathfrak{S} \right] &= \tr_{\text{aux},1}\left[ T \right] S^{(0)}, \\
	& & \tr\left[ \mathrm{L}^\dagger(\mathfrak{S}) \mathfrak{T} \right] &= \tr_{\text{aux},1} \left[ (L^\dagger)^{(1)} T \right].
	\end{aligned}
	\end{equation}
	 We conclude that $(L^\dagger)^{(1)} = S^{(0)} \mathds{1}_{\text{aux},1}$ and thus by setting $\lambda := S^{(0)}$ the dual problem takes the form
	\begin{equation} \begin{aligned}
	& \text{minimize} & & \lambda, \\
	& \text{subject to} & \lambda \, \mathds{1}_{\text{aux},1} &\succeq \mathds{1}_{\text{aux}} \otimes  \tr_{2}[\Omega_\theta(h)].
	\end{aligned} \end{equation}
	We note that the last constraint is equivalent to $\lambda \, \mathds{1}_{1} \succeq  \tr_{2}[\Omega_\theta(h)]$ and the minimal $\lambda$ is equal to $\norm{ 	\tr_{2}[\Omega_\theta(h)] }_{\infty}$. Taking into account the minimization over Hermitian matrices, we find
	\begin{equation}
	J(C_\theta) = \min_h \norm{ 	\tr_{2}[\Omega_\theta(h)] }_{\infty}.
	\end{equation}
	 	
	Using the relation between ensemble decomposition $\lket{C_{\theta,i}}$ and Kraus operators $\hat{C}_{\theta,i}$  	
	 	\begin{align} \label{eq:kraus} \ket{C_{\theta,i}} =  \sum_{k_2=1}^{\dim \h{2}}(\hat{C}_i^T \otimes \mathds{1}_2)( \ket{k_2} \otimes \ket{k_2}), \end{align}
	 we can evaluate for a fixed $h=0$
	\begin{align} \tr_2\left[\Omega_\theta(h=0)\right] =  4 \tr_{2}\left[ \sum_{i} (\lket{\dot{C}_{\theta,i}} \lbra{ \dot{C}_{\theta,i}}\right]^{T} = 	
	4 \sum_{i} \dot{\hat{C}}_{\theta,i}^\dagger \dot{\hat{C}}_{\theta,i}. 	
	 \end{align}
	Taking into account the minimisation over equivalent Kraus operators we get the QFI of quantum channels, first obtained by Fujiwara and Imai in Ref.~\cite{Fujiwara2008}
		\begin{align} J(C_\theta) = 4 \min_{ \{ \hat{C}_{\theta,i} \} } \norm{\sum_{i} \dot{\hat{C}}_{\theta,i}^\dagger \dot{\hat{C}}_{\theta,i}}_{\infty}. \end{align}

	 \end{proof}
	 
	 \subsection{Equivalent decompositions}
	 \label{sec:equivalent_decomp}
	In this subsection, we show how to let the matrix $h$ enter linearly in the dual problem. In Eq.~\eqref{eq:dual_normalised}, we replace $\lambda$ by $4 \lambda$ and obtain the problem
		\begin{equation} \begin{aligned} \label{eq:dual_A}
	& \text{minimise} & & 4 \lambda \\
	& \text{subject to} &  \tr_{2}\left[ S^{(1)} \right] &= \mathds{1}_{1}, \\
	& & \tr_{2k}\left[ S^{(k)} \right] &= \mathds{1}_{2k-1} \otimes S^{(k-1)}, \quad k=2, \ldots, N-1, \\
	& & \mathds{1}_{2N-1} \otimes \lambda S^{(N-1)} &\succeq \frac{1}{4} \tr_{2N}[\Omega_\theta(h)].
	\end{aligned} \end{equation}	
	
	For a given decomposition $\{ \ket{C_{\theta,i}}\}$, we define $\lket{c_{i,m_{2N}}} = \sum_{m_1 \dots m_{2N-1}} \dot{\overline{\tilde{C}}}_{i,\theta}^{m_1 \dots m_{2N-1} m_{2N}} \lket{m_1 \dots m_{2N-1}}$ ($\overline{a}$ denotes the complex conjugate of $a$) and set up
    \begin{align} \label{eq:equivalentKrausOperatorsCondition}
	 		A(h,\lambda, S) = \left( \begin{array} {ccc|c}
	 		&  & &  \lbra{c_{1,1}} \\
	 		& \mathds{1}_{q\cdot d_{2N}} &  &  	\vdots \\
	 		 &  &  &   \lbra{c_{q,d_{2N}}} \\ \hline
	 		 \lket{c_{1,1}} & 
	 		 \dots & 
	 		 \lket{c_{q,d_{2N}}} & 
	 		\mathds{1}_{2N-1} \otimes \lambda S \end{array}\right),
 	\end{align}
	with $q$ the number of components in the ensemble decomposition and $d_{2N} = \dim(\h{2N})$.
	 By Schur's complement condition, positive semidefiniteness of $A$ is equivalent to

	 	\begin{equation} \mathds{1}_{2N-1} \otimes \lambda S \succeq  \sum_{i=1}^q \sum_{m=1}^{d_{2N}} \lket{c_{i,m}}\,\lbra{c_{i,m}} = \frac{1}{4} \tr_{2N}\left[\Omega_\theta(h)\right]. \end{equation}
	 	We can thus replace the last condition in Eq.~\eqref{eq:dual_A} and find 
	 	\begin{equation} \begin{aligned} 
	& \text{minimise} & & 4 \lambda \\
	& \text{subject to} &  \tr_{2}\left[ S^{(1)} \right] &= \mathds{1}_{1}, \\
	& & \tr_{2k}\left[ S^{(k)} \right] &= \mathds{1}_{2k-1} \otimes S^{(k-1)}, \quad k=2, \ldots, N-1, \\
	& & A(h,\lambda, S^{(N-1)}) & \succeq 0.
	\end{aligned} \end{equation}	
	 We find that this agrees with Algorithm~1.	Equation \eqref{eq:ntilde} implies that the variable $h$ enters linearly in Eq.~\eqref{eq:equivalentKrausOperatorsCondition}. Therefore the minimisation over equivalent decompositions is a semidefinite program.

	\subsection{Dual problem in terms of min-entropy}
	\label{sec:dual_minentropy}
	In this subsection, we show the derivation of Theorem 1 [Eq.~(8) in the main text]. 
	We start analogously to subsection \ref{sec:dualproblemN>1} with Eq.~\eqref{eq:primal}, but want to 
	\begin{equation}
		\begin{aligned}  
		& \text{maximise} & & \tr \left[ (\mathds{1}_{2N} \otimes T) \,   (\mathds{1}_{\text{aux}} \otimes \Omega_\theta(h)) \right], \\	
		& \text{subject to} & &  T \in \text{Comb}[ (\emptyset,\h{1}), \ldots, (\h{2N-2},\h{2N-1}\otimes \h{\text{aux}})], 
		\end{aligned}
		\end{equation}	
	that is we do not trace out over $\h{2N}$. 
	We proceed analogously to Eq.~\eqref{eq:l}, but adjust
	\begin{equation}
		L^{(N)} = \frac{1}{d_{2N}} \tr_{2N-1,2N,\text{aux}}[ \mathds{1}_{2N} \otimes T] - \mathds{1}_{2N-2} \otimes T^{(N-1)},
	\end{equation}
	with $d_{2N} = \dim(\h{2N})$. 
	This in turn changes the dual map $\mathsf{L}^\dagger$ (Eq.~\eqref{eq:ldagger}) to
	\begin{equation}
	(L^\dagger)^{(N)} = \frac{1}{d_{2N}} \mathds{1}_{2N-1,2N,\text{aux}} \otimes S^{(N-1)}.
	\end{equation}	
	The dual problem Eq.~\eqref{eq:dual_normalised} is only affected in the last constraint, and becomes
	
	\begin{equation} \begin{aligned} 
	& \text{minimise} & & \lambda \\
	& \text{subject to} & & S \in \text{Comb}[(\h{1},\h{2}), \ldots, (\h{2N-3},\h{2N-2})], \\
	& & & \frac{1}{d_{2N}} \, \mathds{1}_{2N-1,2N} \otimes \lambda S  \succeq \Omega_\theta(h).
	\end{aligned} \end{equation}	
	
	Comparing with the definition of the conditional min-entropy, we see that the minimal $\lambda$ satisfying the constraints is equal to
		\begin{equation}
		d_{2N} \, 2^{-H_{\text{min}}( N | [N-1])_{\Omega_\theta(h)}},
		\end{equation}
	which agrees with Theorem 1.

\section{Frequency estimation under different types of noise}
	\label{sec:collision}
	In this section, we provide more details about the collision model used for the numerical implementation. In Subsection~\ref{sec:scenarios}, we elaborate on the different scenarios (non-Markovian vs. Markovian, with vs. without feedback), in Subsection~\ref{sec:uint}, we consider different interaction unitaries, in Subsection~\ref{sec:results}, we give some numerical results.

	 \subsection{Scenarios}
	 \label{sec:scenarios}
	 In the following paragraph, we find a comb depending on how we control the system or whether we assume a Markovian or non-Markovian process. The different scenarios are depicted in Figure \ref{fig:scenarios} and we assume that we are given the single-step unitary $U = U_{\text{int}}(\tau) \, (\mathds{1}_E \otimes e^{-\im H t})$.
	 
	 \subsubsection{Non-Markovian case with control}
	  We can depict the scenario where we are allowed to perform control operations between the interactions in Figure \ref{fig:scenarios}(a). The resulting Choi operator of the comb is obtained by tensoring $N$ times the system's degrees of freedom and perform matrix multiplication on the environmental degrees of freedom, resulting in an operator $U_{\text{tot}} \in \mathcal{L}(\h{E} \otimes \bigotimes_{i=1}^{2N} \h{i})$. We initialise the first environment in $\ket{0}$ and trace out over the last environment $\h{E_{2N}}$. In order to obtain a vector, we flatten this matrix to a vector and denote this operation by $\mathrm{vec}$. This results in the ensemble decomposition
		\begin{align} \label{eq:kraus_interpretation} \ket{C_{\theta,i}} = \mathrm{vec}(\mel{i}{U_{\text{tot}}}{0}), \end{align}
	where $\ket{i}$ denotes an orthonormal basis of the last environmental system $\h{E_{2N}}$.
		 We note that, if we include the environment, the evolution can be described by a unitary, and the ensemble decomposition results of our ignorance about the state of the environment. Note that the above expression	resembles the expression for the Kraus operator $\hat{E}_i$ of a channel $\mathcal{E}(\rho) = \sum_{i} \hat{E}_i \rho \hat{E}_i^\dagger$, which can be obtained by its Stinespring dilation unitary $U$ in the following way: $\hat{E}_i = \mel{i}{U}{0}$, where $\ket{0}$ denotes the initial state of the environment and $\ket{i}$ is an orthonormal basis of the environment \cite{Nielsen2000}.
		 %[chapter 8.2.3]
	 
	 \subsubsection{Non-Markovian case without control} 
	 If we  cannot perform any control operations, we have the scenario depicted in Figure \ref{fig:scenarios}(b). The total unitary $U_{\text{tot}} = U^N$ is the $N$-fold concatenation of $U$ resulting in an ensemble decomposition
	 \begin{equation}
	 \ket{C_{\theta,i}} = \mathrm{vec}(\mel{i}{U^N}{0}), 
	 \end{equation}
	where as before we flatten the matrix to a vector.
	
		\begin{figure}
	\centering
	\includegraphics[width=.7\columnwidth]{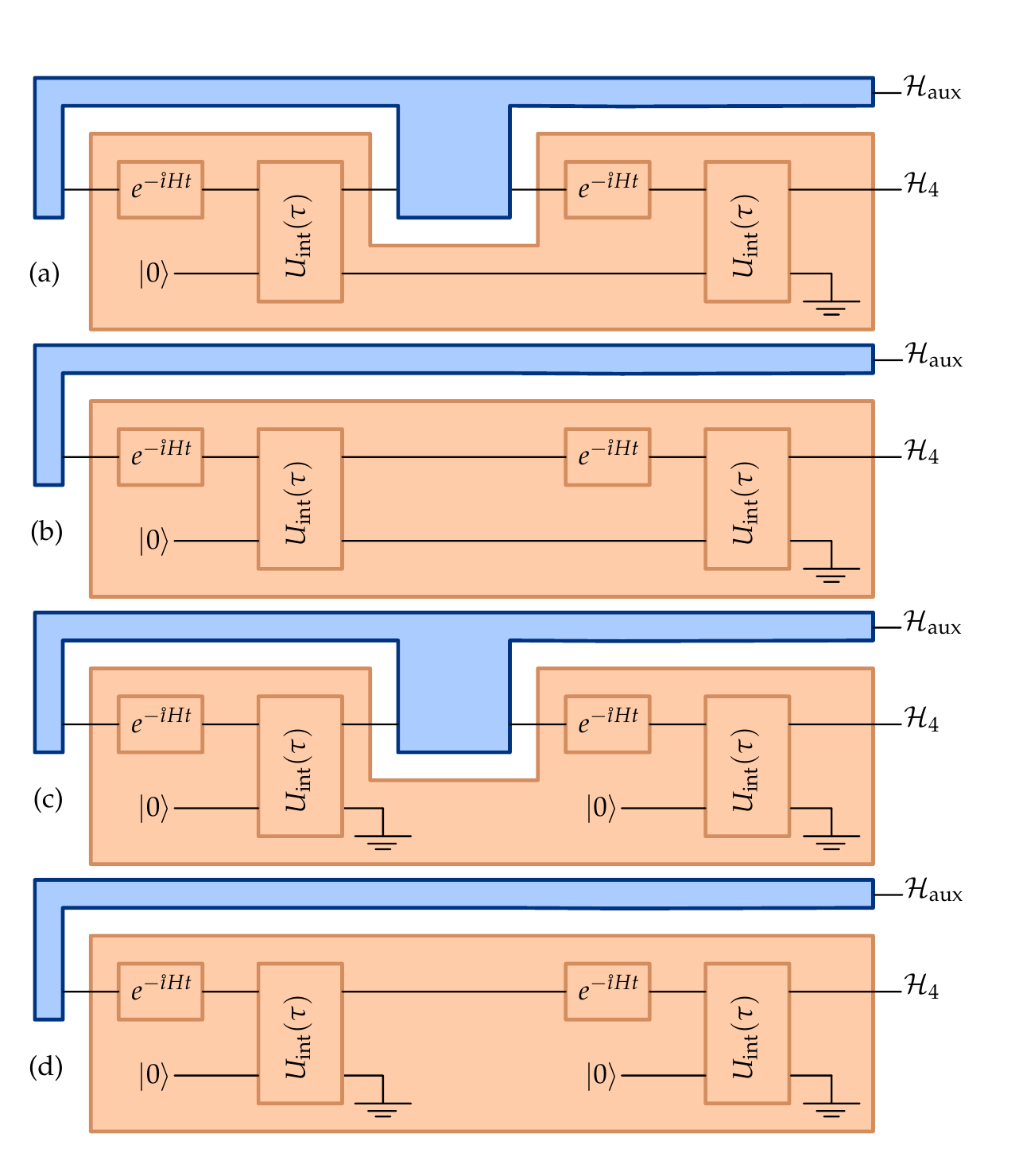}
	\caption{Collision model with $N=2$ interactions. We distinguish between non-Markovian processes (a, b) and Markov processes (c, d). In addition we distinguish between operations with control (a, c) and without control (b, d).}
	\label{fig:scenarios}
	\end{figure}

	\subsubsection{Markovian case with control}
	We assume the environment to have no memory. Hence, after each interaction, we trace over the environment and reinitialise it in the next step in the ground state, see Figure \ref{fig:scenarios}(c) for the $N=2$ case. For each interaction step $k=1, \ldots, N$, we obtain the components $\ket{C_{\theta,i_k}} = \mathrm{vec}(\mel{i_k}{U}{0})$ in the ensemble decomposition with $\ket{i_k}$ an orthonormal basis of the environment $\h{E_k}$. In the end, we tensor these components to obtain the total ensemble decomposition
	\begin{equation}
		\lket{C_{\theta,\vec{i} = (i_1, \ldots, i_{N})}} = \bigotimes_{k=1}^{N} \ket{C_{\theta,i_k}}.
\end{equation}

	\subsubsection{Markovian case without control} 
	In the Markovian case without control, we trace out the environment after each interaction and reinitialise it in the ground state and thus obtain components $U_{\theta,i_k} = \mel{i_k}{U}{0}$ for each interaction $k=1,\ldots, N$. In addition, we do not perform any operation on the system's side. This scenario is depicted in Figure \ref{fig:scenarios}(d). We obtain the resulting ensemble decomposition by concatenating the components $C_{\theta,i_k}$:
	\begin{equation}
	\lket{C_{\theta,\vec{i}=(i_1, \ldots, i_N)}} = \mathrm{vec}\left(\prod_{k=1}^N U_{\theta,i_k}\right).
	\end{equation}

	\subsection{Different interactions} 
	\label{sec:uint}
	We give some examples \cite{Ziman2005} for unitary interactions $U = U_{\text{int}}(\tau) \, (\mathds{1}_E \otimes e^{-\im H t})$.
	\subsubsection{Partial SWAP}
	Consider a qubit in $\h{S}$ interacting with a partial swap with a two-dimensional environment $\h{E}$
	%[chapter 4]
	: In the basis $\{\ket{00}_{ES}$, $\ket{01}_{ES}$, $\ket{10}_{ES}$, $\ket{11}_{ES}\}$ the swap is given as
		\begin{align} \label{eq:swap} \text{SWAP} = \begin{pmatrix} 1 & 0 & 0 & 0 \\ 0 & 0  & 1 & 0 \\ 0 & 1 & 0 & 0 \\ 0 & 0 & 0 & 1 \end{pmatrix}. \end{align}
		
	If the system interacts with the environment with interaction strength $g$ and for time $\tau$ under this SWAP, the evolution is $U_{\text{SWAP}}(\tau) = e^{-\im g \tau \text{SWAP}}$.
%		\begin{align}
%		U_{\text{SWAP}}(\tau) = \begin{pmatrix} e^{-\im g \tau} & 0 & 0 & 0 \\ 0 & \cos(g \tau) & \im \sin(g \tau) & 0 \\ 0 & \im \sin(g \tau) & \cos(g \tau) & 0 \\ 0 & 0 & 0 & e^{-\im g \tau} \end{pmatrix}
%		\end{align}		
%	
	 We assume that the qubit's evolution is governed by a phase shift $H = \omega \ket{1}\bra{1}$, thus in the same basis as above, the time evolution of the joint system is
		\begin{equation} U = 
		\begin{pmatrix} e^{-\im g \tau} & 0 & 0 & 0 \\ 0 &  \cos(g \tau)e^{-\im \omega t} & -\im \sin(g \tau) & 0 \\ 0 & -\im \sin(g \tau) e^{-\im \omega t}  & \cos(g \tau) & 0 \\ 0 & 0 & 0 & e^{-\im (\omega t+g \tau)} \end{pmatrix}.				\end{equation}
	 
%	 The plots for different $t_{\text{tot}}$ and $N=3$ can be seen in Figure \ref{fig:partialSwapPlot}. 
%	
%	\begin{figure}
%	\includegraphics[width=.5\textwidth]{../Simulation/Plots/UsefulPlots/swap-identity-N-3_2020-12-01_09-36-57-998.pdf}
%	%\includegraphics[width=.5\textwidth]{../Simulation/Plots/UsefulPlots/swap-sqrt-N-3_2020-11-17_15-59-33-810.eps}
%	\caption{\qfi\, of partial SWAP with $N=3$, $\omega = 0.05 \cdot 2\pi$, $g=1$ as a function of the total time $t_{tot}$. We choose the interaction time $\tau$ is chosen equal to the evolution time $t$.}
%	\label{fig:partialSwapPlot}
%	\end{figure}
	
\subsubsection{Partial CNOT with environment as control}	
	We consider a partial CNOT as interaction. If we consider the environment as the control and the system as the target, the CNOT acts on the total system as
		\begin{equation} \text{CNOT}_E = \begin{pmatrix} 1 & 0 & 0 & 0 \\
			0 & 1 & 0 & 0 \\
			0 & 0 & 0 & 1 \\
			0 & 0 & 1 & 0 \end{pmatrix}. \end{equation}
	in the same basis  as above. The corresponding unitary is $U_{\text{CNOT}_E}(\tau) = e^{-\im  g \tau \text{CNOT}_E}$ where $ g$ is the interaction strength and $\tau$ the interaction time. We assume that the system undergoes a phase shift $H = \omega \ket{1}\bra{1}$ for time $t$ before the partial CNOT is applied. The resulting unitary is
		\begin{equation} U_{\text{CNOT}_E}(\tau) (\mathds{1}_E \otimes e^{-\im H t}) = 
		\begin{pmatrix} e^{-\im  g \tau} & 0 & 0 & 0 \\
		0 & e^{-\im ( g \tau + \omega t)} & 0 & 0 \\
		0 & 0 & \cos( g \tau) &  -\im \sin(g \tau) e^{-\im \omega t} \\
		0 & 0 & -\im \sin(g \tau) & \cos(g \tau) e^{-\im \omega t} \end{pmatrix}	. \end{equation}
	If we initialise the environment in $\ket{0}_{E_1}$, we note that the partial CNOT does not have any effect. That is why we choose to initialise in $\ket{+}_{E_1} = \frac{1}{\sqrt{2}} (\ket{0}_{E_1} + \ket{1}_{E_1})$. In order to be able to use the same setup as before, we set 
	\begin{equation} U = (U_{\text{basis}} \otimes \mathds{1}_S)^\dagger \, U_{\text{CNOT}_E}(\tau) \, (\mathds{1}_E \otimes e^{-\im H t}) \, (U_{\text{basis}} \otimes \mathds{1}_S), \end{equation}
	with $U_{\text{basis}}$ given by
		\begin{equation} U_{\text{basis}}  = \frac{1}{\sqrt{2}} \begin{pmatrix}
		1 & 1 \\ 1 & -1 \end{pmatrix}. \end{equation}
	 
%	 Some plots are depicted in Figure \ref{fig:partialCNOTPlot}.
%	
%	\begin{figure}
%	\includegraphics[width=.5\textwidth]{../Simulation/Plots/UsefulPlots/cnot_environment-identity-N-3_2020-12-01_08-31-10-654.pdf}
%	%\includegraphics[width=.5\textwidth]{../Simulation/Plots/UsefulPlots/cnot_environment-sqrt-N-3_2020-11-17_16-10-15-174.eps}
%	\caption{\qfi of partial CNOT with environment acting as the control and with $N=3$, $\omega = 0.05 \cdot 2\pi$, $g = 1$ as a function of the total time $t_{tot}$. The interaction time $\tau$ is chosen equal to the evolution time $t$.}
%	\label{fig:partialCNOTPlot}
%	\end{figure}	
	
	\subsubsection{Partial CNOT with system as control} 
	If we assume the system to control the CNOT and the environment to be the target, we start with
		\begin{equation} \text{CNOT}_{S} = \begin{pmatrix} 1 & 0 & 0 & 0 \\
		0 & 0 & 0 & 1 \\
		0 & 0 & 1 & 0 \\
		0 & 1 & 0 & 0 \end{pmatrix}, \end{equation}
	where we used the same basis as before. Analogously to the previous case, we then find $U = e^{-\im g \tau \text{CNOT}_{S}} (\mathds{1}_E \otimes e^{-\im H t}_S)$ to be
		\begin{equation} U = \begin{pmatrix}
		e^{-\im g \tau} & 0 & 0 & 0 \\
		0 &  \cos(g \tau)e^{-\im \omega t}  & 0 & -\im \sin(g \tau) e^{-\im \omega t} \\
		0 & 0 & e^{-\im g \tau} & 0 \\
		0 & -\im \sin(g \tau) e^{-\im \omega t} & 0 & \cos(g \tau) e^{-\im \omega t}
		\end{pmatrix}. \end{equation}

	\subsubsection{Bitflip channel}
	Consider a system-environment interaction $H_{\text{int}} = g X \otimes X$ with $g$ the interaction strength and $X$ the Pauli $X$ matrix. This interaction is a good approximation for situations in which the coupling to a single spin from the bath dominates over other interactions \cite{Oreshkov2007, Krovi2007}.	
	The environment starts in the completely mixed state $\rho_{E,0} = \frac{1}{2} \mathds{1}_E$, which corresponds to an equilibrium state at high temperature. Furthermore we assume that the system's initial state is uncorrelated with the environment: $\rho_0 = \rho_{S,0} \otimes \rho_{E,0}$. Note that $e^{-\im g \tau \, X \otimes X} = \cos(g \tau) \mathds{1}_{SE} - \im \sin(g \tau) X \otimes X$ and the unitary $U = U_{\text{int}}(\tau) \, (\mathds{1}_E \otimes e^{-\im H t})$ is given by
	
	\begin{equation} U = \begin{pmatrix}
		\cos(g \tau) & 0 & 0 & -\im \sin(g \tau) e^{-\im \omega t}  \\
		0 &  \cos(g \tau) e^{-\im \omega t} & -\im \sin(g \tau)  & 0 \\
		0 & -\im \sin(g \tau) e^{-\im \omega t} & \cos(g \tau) & 0 \\
		-\im \sin(g \tau) & 0 & 0 & \cos(g \tau) e^{-\im \omega t}
		\end{pmatrix}.
	\end{equation}
		
	 We observe that the state after an interaction corresponds to
		\begin{equation}
		\begin{aligned}
		& & \rho_{SE}(\tau) &= \left( \cos^2(g \tau) \rho_{S,0} + \sin^2(g \tau) \, X \rho_{S,0} X \right) \otimes \frac{1}{2} \mathds{1}_E \\
		& & &+ \im \cos(g \tau) \sin(g \tau) ( \rho_{S,0} X + X \rho_{S,0}) \otimes \frac{1}{2} X.
		\end{aligned}
		\end{equation}
	The dynamics reduces on the system's space to
		\begin{equation}
		\rho_S(\tau) = \tr_E\left[ \rho_{SE}(\tau) \right] = \cos^2(g \tau) \rho_{S,0} + \sin^2(g \tau) \, X \rho_{S,0} X.
		\end{equation}
	Setting $p = \sin^2(g \tau)$, we recover the bit flip channel
	\begin{equation}
	\rho_S = (1-p) \rho_{S,0} + p X \rho_{S,0} X.
\end{equation}

	\subsection{Results}
	\label{sec:results}	
	We implement the semidefinite program using CVX \cite{CVX}.
	For the calculation, we choose $N=2,3,4$, the frequency $\omega = \frac{\pi}{10}$, the total time $t_{\text{tot}} \in [0,21]$ and the interaction strength $g = 1$. This ensures that $ \omega \cdot t_{\text{tot}}$ spans $[0, 2\pi]$. By repeating the simulation for different values of $\omega$, we see that the QFI and the optimal tester do not depend on $\omega$. This allows to probe the comb independently of the true value of $\omega$. We choose the sampling and interaction time of a single evolution to be $t = \tau = \frac{t_{\text{tot}}}{N}$.
	%The dimension of the purifying system is chosen to be $\dim \h{\widetilde{1}} = 2$. Note that in principle larger dimensions could lead to an improvement \cite{Pollock2018} but the memory requirements become exceedingly large.
	Some of the obtained plots are depicted in Figure \ref{fig:N2}. The noiseless limit $J_Q(t_{\text{tot}}) = t_{\text{tot}}^2$ is realised by setting the interaction strength $g = 0$. We note that the non-Markovian case with control achieves the highest QFI for the partial SWAP. For the partial CNOT with the environment as control and the bitflip channel, it even achieves the noiseless scaling $J_Q = t_{\text{tot}}^2$. For $g\tau$ a multiple of $\pi$, the system-environment interaction becomes trivial and we recover the noiseless scaling for all scenarios. For $g \tau = \frac{2k+1}{2} \pi$ with $k$ an integer, the interaction between system and environment is maximal; e.g., for the swap, all the information about $\theta$ is swapped to the environment and erased in the Markovian case. 

	\begin{figure}[h!]
	\begin{center}
	\includegraphics[width=1\textwidth]{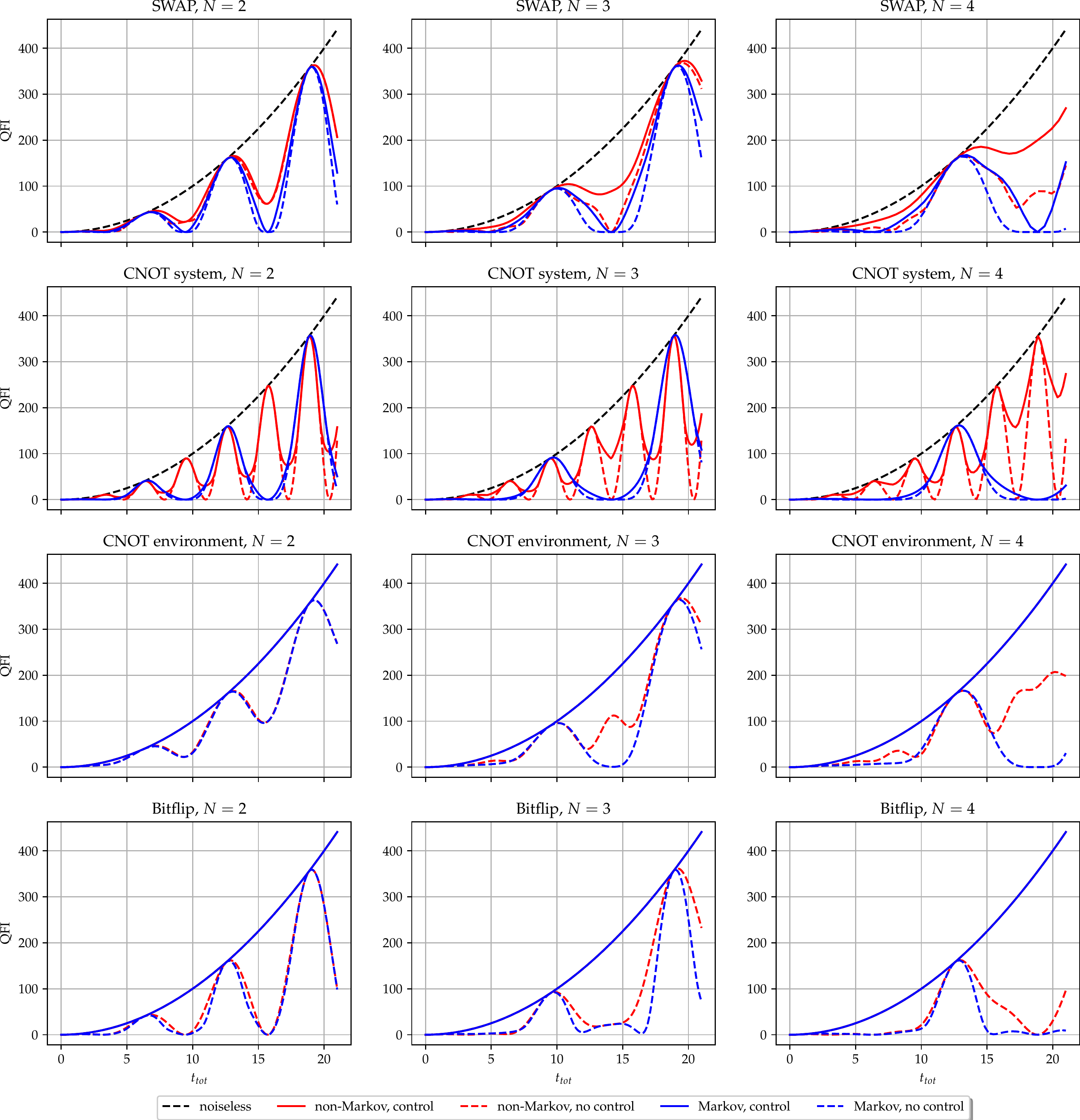}
	\caption{Quantum Fisher information as a function of the total evolution time. The interaction time $\tau$ is chosen equal to the evolution time $t$. For the CNOT environment and the bitflip, the blue and red solid lines overlap with the black dashed line.}
	\label{fig:N2}
	\end{center}
	\end{figure}	
	
	Since the dual problem finds the optimal $h$ for $\Omega_\theta(h)$, we can insert it in the primal problem to find the optimal probe $T$. The optimal probes are diagonal. The data can be found on a repository \cite{Altherr2021}.
%	For the $N=2$ case, we have
%		\begin{equation} \label{eq:T_opt} 
%		\begin{aligned} 
%		& &T_{\text{SWAP}} &=  \text{diag}\left(1,0,1/2,1/2,0,0,0,0\right), \\ 
%		& &T_{\text{CNOT,env}} &=  \text{diag}\left(1/2, 0, 1/2, 0, 0, 1/2, 0, 1/2 \right), \\
%		& &T_{\text{CNOT,sys}} &=  \text{diag}\left(0, 0, 0, 0, 1/2, 1/2, 0, 1 \right), \\
%		& &T_{\text{bitflip}} &= \text{diag}\left(1/2, 0 , 1/2, 0, 0, 1/2, 0, 1/2 \right).		
%		\end{aligned}
%		\end{equation}
		Note that, e.g., for $N=2$ the probe is an element in $T \in  \text{Comb}[(\emptyset, \h{1}), (\h{2},\h{3})]$, since the auxiliary space $\h{\text{aux}}$ does not enter the SDP. In this work, we only exactly calculate the QFI of a generic quantum $N$-comb for $N \leq 4$, and calculation regarding the $N>5$ cases is likely to require new numerical methods and/or approximations.
		
\section{QFI for adaptive channel estimation}
	An interesting question in channel estimation arises by considering $N$ copies of a parametrised quantum channel $\mathcal{E}_\theta$. Here we show that the QFI corresponding to the optimal adaptive strategy for this task can be readily evaluated using our result. 

	The asymptotic limit of the QFI for $N \rightarrow \infty$ copies has been treated in Ref.~\cite{Fujiwara2008, Demkowicz2012}. We can use our framework to model this question by considering the Choi operator $E_\theta$ of that quantum channel with decomposition $E_\theta = \sum_i \ket{E_{\theta,i}}\bra{E_{\theta,i}}$. We consider $N$ copies that act on subsequent Hilbert spaces $E_\theta \in \text{Comb}[(\h{2k-1},\h{2k})]$ for $k=1,\ldots, N$. The resulting comb of these $N$ channels is the tensor product of the individual channels $C_\theta := E_\theta^{\otimes N}$ and we find an ensemble decomposition $C_\theta = \sum_{\vec{i}} \lket{C_{\theta,\vec{i}}} \lbra{C_{\theta,\vec{i}}}$ (we use a vectorial index $\vec{i} = (i_1, \ldots, i_N)$) with
	\begin{align}
	\lket{C_{\theta,\vec{i} = (i_1, \ldots, i_N)}} = \bigotimes_{k=1}^N \lket{E_{\theta,i_k}}.
	\end{align}
	We determine the performance operator of $C_\theta$ by computing the derivatives of $\lket{C_{\theta,\vec{i}}}$:
	\begin{align}
	\lket{\dot{C}_{\theta,\vec{i}}} = 
	\sum_{k=1}^N \lket{E_{\theta,i_1}} \dots \lket{E_{\theta,i_{k-1}}} \lket{\dot{E}_{\theta,i_{k}}} \lket{E_{\theta,i_{k+1}}}  \dots \lket{E_{\theta,i_N}}
	%\sum_{N_1 + N_2 + 1 = N} \left(\bigotimes_{k=1}^{N_1} \lket{E_{\theta,i_k}}\right) \lket{\dot{E}_{\theta,i_{N_1+1}}} \left( \bigotimes_{k=1}^{N_2}\lket{E_{\theta,i_k}} \right).
	\end{align}		
	We insert this derivative into
	\begin{align}
	\Omega_\theta^{N \text{ copies}} := 4 \sum_{\vec{i}} \lket{\dot{C}_{\theta,\vec{i}}} \lbra{\dot{C}_{\theta,\vec{i}}}
	\end{align}
	and note that $\Omega_\theta^{N \text{ copies}}$ consists of the following terms
	\begin{align}
	&\sum_{i_k} \lket{\dot{E}_{\theta,i_k}} \lbra{\dot{E}_{\theta,i_k}} = \tfrac{1}{4} \Omega_\theta, \\
	&\sum_{i_k} \lket{\dot{E}_{\theta,i_k}} \lbra{E_{\theta,i_k}} =: \tfrac{1}{2}(\dot{E}E)_\theta, \\
	&\sum_{i_k} \lket{E_{\theta,i_k}} \lbra{\dot{E}_{\theta,i_k}} =: \tfrac{1}{2} (E\dot{E})_\theta, \\
	&\sum_{i_k} \lket{E_{\theta,i_k}} \lbra{E_{\theta,i_k}} = E_\theta,	
	\end{align}
	where $\Omega_\theta$ is the performance operator of the channel $E_\theta$. More explicitly,
	\begin{align}
	\Omega_\theta^{N \text{ copies}} =
	\sum_{i+j+1=N} E_\theta^{\otimes i} \Omega_\theta E_\theta^{\otimes j} +
	\sum_{i+j+k+2=N} E_\theta^{\otimes i} [ (\dot{E}E_\theta) \, E_\theta^{\otimes j} \,  (E\dot{E})_\theta + (E\dot{E})_\theta \, E_\theta^{\otimes j} \, (\dot{E}E)_\theta)] E_\theta^{\otimes k}.
\end{align}		
	Finally, according to Theorem 1 of the main text, the QFI of the optimal adaptive strategy is determined by $\Omega_\theta^{N \text{ copies}}$ (or, more explicitly, the conditional comb min-entropy of $\Omega_\theta^{N \text{ copies}}$). The optimal adaptive strategy of channel estimation can therefore be determined.
\end{widetext}

\end{document}